\documentclass[10pt,onecolumn]{IEEEtran}

\usepackage{amsmath,amssymb,amsfonts,graphicx}
\usepackage{latexsym}
\usepackage{amsthm}
\usepackage{latexsym}
\usepackage{amsthm,cite}
\usepackage{url}
\usepackage{algorithm}
\usepackage{algpseudocode}

\usepackage{caption}
\usepackage{subcaption}

\newcommand{\database}{{\mathcal{D}^n}}    
\newcommand{\noise}{{X}}    
\newcommand{\KM}{{\mathcal{K}}}    
\newcommand{\e}{{\epsilon}}    
\newcommand{\dlt}{{\delta}}    
\newcommand{\loss}{{\mathcal{L}}}    
\newcommand{\D}{{\Delta}}    
\newcommand{\p}{{\mathcal{P}}}  

\newcommand{\R}{{\mathbb{R}}}  
\newcommand{\N}{{\mathbb{N}}}  
\newcommand{\Z}{{\mathbb{Z}}}  

\newtheorem{theorem}{Theorem}

\newtheorem{definition}{Definition}
\newtheorem{corollary}[theorem]{Corollary}

\begin{document}

\title{Optimal Noise Adding Mechanisms for Approximate Differential Privacy}

\author{
\authorblockN{Quan Geng, and Pramod Viswanath}\\
\authorblockA{Coordinated Science Laboratory and Dept. of ECE \\
University of Illinois, Urbana-Champaign, IL 61801 \\
Email: \{geng5, pramodv\}@illinois.edu} }

\maketitle

\begin{abstract}

We study the (nearly) optimal mechanisms in $(\epsilon,\delta)$-approximate differential privacy for integer-valued query functions and vector-valued (histogram-like) query functions under a utility-maximization/cost-minimization framework. We characterize the tradeoff between $\epsilon$ and $\delta$ in utility and privacy analysis for histogram-like query functions ($\ell^1$ sensitivity), and show that the $(\epsilon,\delta)$-differential privacy is a framework not much more general than the $(\epsilon,0)$-differential privacy and $(0,\delta)$-differential privacy in the context of $\ell^1$ and $\ell^2$ cost functions, i.e., minimum expected noise magnitude and noise power. In the same context of $\ell^1$ and $\ell^2$ cost functions,  we show the near-optimality of uniform noise mechanism and discrete Laplacian mechanism in the high privacy regime (as $(\epsilon,\delta) \to (0,0)$). We conclude that in $(\epsilon,\delta)$-differential privacy,  the optimal noise magnitude and noise power are $\Theta(\min(\frac{1}{\epsilon},\frac{1}{\delta}))$ and $\Theta(\min(\frac{1}{\epsilon^2},\frac{1}{\delta^2}))$, respectively, in the high privacy regime.

\end{abstract}

\section{Introduction} \label{sec:intro}


Differential privacy is a framework to quantify to what extent  individual privacy in a statistical database is preserved while releasing useful statistical information about the database \cite{DMNS06}. The basic idea of differential privacy is that the presence of any individual data in the database should not affect the final released statistical information   significantly, and thus it can give strong privacy guarantees against an adversary with arbitrary auxiliary information. For more background  and motivation of differential privacy, we refer the readers to the survey \cite{DPsurvey}.

The standard approach to preserve $\epsilon$-differential privacy for real-valued query function is to perturb the query output by adding random noise with Laplacian distribution. Recently, Geng and Viswanath \cite{GV13} show that under a general utility-maximization framework, for single real-valued query function, the optimal $\e$-differentially private mechanism is the staircase mechanism, which adds noise with staircase distribution to the query output. The optimality of the staircase mechanism is extended to the multidimensional setting for histogram-like functions in \cite{Geng13}, where the sensitivity of the query functions is defined using the $\ell^1$ metric as in \cite{DMNS06}. A relaxed notion of privacy, $(\e,\delta)$-differential privacy, was introduced by Dwork et al. \cite{DKMMN06}, and the standard approach to preserving   $(\e,\delta)$-differential privacy is to add Gaussian noise to the query output.

In this work, we study the (nearly) optimal mechanisms in $(\e,\delta)$-differential privacy for integer-valued query functions and vector-valued (histogram-like) query functions under a utility-maximization/cost-minimization framework, and characterize the tradeoff between $\epsilon$ and $\delta$ in utility and privacy analysis. $(\e,\delta)$-differential privacy is a relaxed notion of privacy, compared to the standard $\e$-differential privacy introduced in \cite{DMNS06}. $(\e,\delta)$-differential privacy includes as special cases:
\begin{itemize}
\item  $(\e,0)$-differential privacy;  in this standard setting, the optimal mechanism for a general cost minimization framework is the {\em staircase} mechanism as shown in \cite{GV13} and \cite{Geng13}. In the high privacy regime, the standard discrete Laplacian mechanism too performs well.
\item $(0,\delta)$-differential privacy; this setting requires that the total variation of the conditional probability distributions of the query output for neighboring datasets should be bounded by $\delta$. In this paper we show that the uniform noise distribution is  near-optimal in the $(0,\delta)$-differential privacy setting for a general class of cost functions.
\end{itemize}

While the $(\e,\delta)$-differential privacy setting is more general than the two special cases -- $(\e,0)$ and $(0,\delta)$-differential privacy -- our main result in this work is to show that it is only more general by {\em very little}; this is done in the context of $\ell^1$ and $\ell^2$ cost functions. We show the near-optimality of uniform noise and discrete Laplacian mechanisms in the high privacy regime (as $(\e,\delta) \to (0,0)$) for $\ell^1$ and $\ell^2$ cost functions.

Our result is a sharp departure from the setting of $\ell^{\infty}$ sensitivity (modeling adaptive query compositions) where the notion of $(\e,\delta)$-approximate differential privacy provides significant variance reductions (in the dimension of the query output), as compared to the standard $(\e,0)$-differential privacy \cite{DMNS06,DL09,DRV10}. Our main result shows that such gains are not available in the $\ell^1$ sensitivity model -- in fact approximate differential privacy in the usual regime ($\delta << \e$) is nearly the same (up to constants, in added noise magnitude and variance) as regular differential privacy. For completeness, we consider all relationships between $\e$ and $\delta$ in this paper.

The near-optimality of the two mechanisms (designed for the special cases of $(\epsilon,0)$ and $(0,\delta)$ differential privacy settings) is proved by demonstrating a uniform bound on the ratio between the costs of these two mechanisms and that of the optimal cost in the $(\epsilon,\delta)$ differential privacy setting in the high privacy regime, i.e.,  as $(\e,\delta) \to (0,0)$ for $\ell^1$ and $\ell^2$ cost functions.


\subsection{Summary of Our Results}

In this work we consider a very general model for integer-valued and vector-valued (histogram-like) query functions. Unlike previous works on $(\e,\delta)$-differential privacy (e.g., \cite{NTZ12}, \cite{KRSU10}, \cite{De12}), we impose no assumptions on the dataset model and properties of the query functions other than the global sensitivity, which is defined using the $\ell^1$ metric. We implicitly assume that the local sensitivity is equal to the global sensitivity. Due to the optimality of query-output independent perturbation as shown in \cite{GV13}, we consider query-output independent perturbation mechanisms.

We summarize our results in the following.  Let $V_{LB}$ denote the lower bound we derived for the cost under differential privacy constraint. Let  $V_{UB}^{\mbox{Lap}} $  and $V_{UB}^{\mbox{uniform}}$ denote the upper bounds for the cost achieved by discrete Laplacian mechanism and uniform noise mechanism. In this work, we show that
\begin{itemize}
	\item For integer-valued query functions,
	\begin{itemize}
		\item for $(0,\delta)$-differential privacy with the global sensitivity $\D = 1$, the uniform noise mechanism is optimal for all generic cost funtions,
		\item for $(0,\delta)$-differential privacy with arbitrary global sensitivity $\D$, $\lim_{\delta \to 0} \frac{V_{UB}^{\mbox{uniform}}}{V_{LB}} = 1$ for $\ell^1$ and $\ell^2$ cost functions,
		\item for $(\e,\delta)$-differential privacy with $\ell^1$ and $\ell^2$ cost functions, $\lim_{(\e,\delta) \to (0,0)} \frac{\min(V_{UB}^{\mbox{Lap}},V_{UB}^{\mbox{uniform}}  ) }{V_{LB}} \le C$ for some numerical constant $C$.
	\end{itemize}

	\item For vector-valued (histogram-like) query functions,
	\begin{itemize}
		\item  for $(0,\delta)$-differential privacy with the global sensitivity $\D = 1$, the multi-dimensional uniform noise mechanism  is optimal for $\ell^1$ and $\ell^2$ cost funtions,
		\item for $(0,\delta)$-differential privacy with arbitrary global sensitivity $\D$, $\lim_{\delta \to 0} \frac{V_{UB}^{\mbox{uniform}}}{V_{LB}} = 1$ for $\ell^1$ and $\ell^2$ cost functions,
		\item for $(\e,\delta)$-differential privacy with $\ell^1$ and $\ell^2$ cost functions, $\lim_{(\e,\delta) \to (0,0)} \frac{\min(V_{UB}^{\mbox{Lap}},V_{UB}^{\mbox{uniform}}  ) }{V_{LB}} \le C$ for some numerical constant $C$, which is independent of the dimension of the query function.
	\end{itemize}
\end{itemize}

We conclude that in $(\epsilon,\delta)$-differential privacy, the optimal noise magnitude and noise power are $\Theta(\min(\frac{1}{\epsilon},\frac{1}{\delta}))$ and $\Theta(\min(\frac{1}{\epsilon^2},\frac{1}{\delta^2}))$, respectively, in the high privacy regime, and naturally, the total cost grow linearly in terms of the dimension of the query output.

\subsection{Related Work}

Dwork et. al. \cite{DMNS06} introduce $\e$-differential privacy and show that the Laplacian mechanism, which perturbs the query output by adding random noise with Laplace distribution proportional to the global sensitivity of the query function, can preserve $\e$-differential privacy. In \cite{DMNS06}, it is shown that for histogram-like query functions, where the query output has multiple components and the global sensitivity is defined using the $\ell^1$ metric, one can perturb each component independently by adding the Laplacian noise to preserve $\e$-differential privacy.

Nissim, Raskhodnikova and Smith \cite{NRS07} show that for certain nonlinear query functions, one can improve the accuracy by adding data-dependent noise calibrated to the smooth sensitivity of the query function, which is based on the local sensitivity of the query function.
McSherry and Talwar \cite{McSherry07} introduce the \emph{exponential mechanism} to preserve $\e$-differential privacy for general query functions in an abstract setting, where the query function may not be real-valued.
Dwork et. al. \cite{DKMMN06} introduce $(\e,\delta)$-differential privacy and show that adding random noise with Gaussian distribution can preserve $(\e,\delta)$-differential privacy for real-valued query function. Hall, Rinaldo, and Wasserman \cite{Hall13} study how to preserve $(\e,\delta)$-differential privacy for releasing (infinite dimensional) functions, and show that adding Gaussian process noise to the released function can preserve  $(\e,\delta)$-differential privacy.
Kasiviswanathan and Smith \cite{Smith08} study $(\e,\delta)$-semantic privacy under a Bayesian framework. Chaudhuri and Mishra \cite{Chaudhuri06}, and Machanavajjhala et. al. \cite{Kifer08} propose different variants of the standard $\e,\delta$-differential privacy.

Ghosh, Roughgarden, and Sundararajan  \cite{Ghosh09} show that for a single count query with sensitivity $\D = 1$, for a general class of utility functions, to minimize the expected cost under a Bayesian framework the optimal mechanism to preserve $\e$-differential privacy is the geometric mechanism, which adds noise with geometric distribution. Brenner and Nissim  \cite{Nissim10}  show that for general query functions no universally optimal mechanisms exist. Gupte and Sundararajan \cite{minimax10} derive the optimal noise probability distributions for a single count query with sensitivity $\D = 1$ for minimax (risk-averse) users. \cite{minimax10} shows that although there is no universally optimal solution to the minimax optimization problem in \cite{minimax10} for a general class of cost functions, each solution (corresponding to different cost functions) can be derived from the same geometric mechanism by randomly remapping. Geng and Viswanath \cite{GV13} generalize the results of \cite{Ghosh09} and \cite{minimax10} to real-valued (and integer-valued) query functions with arbitrary sensitivity, and show that the optimal query-output independent perturbation mechanism is the staircase mechanism, which adds noise with a staircase-shaped probability density function (or probability mass function for integer-valued query function) to the query output. The optimality of the staircase mechanism is extended to the multidimensional setting for histogram-like functions in \cite{Geng13}, where the sensitivity of the query functions is defined using the $\ell^1$ metric as in \cite{DMNS06}.

Differential privacy for histogram query functions has been widely studied in the literature, e.g.,  \cite{Fang12,Hay10,geometry,Xu2012,Li10,NTZ12}, and many existing works use the Laplacian mechanism as the basic tool. For instance, Li et al. \cite{Li10} introduce the matrix mechanism to answer batches of linear queries over a histogram in a differentially private way with good accuracy guarantees. Their approach is that instead of adding Laplacian noise to the workload query output directly, the matrix mechanism will design an observation matrix which is the input to the database,  from perturbed output (using the standard Laplace mechanism)  estimate the histogram itself, and then compute the query output directly.  \cite{Li10} shows that this two-stage process will preserve differential privacy and increase the accuracy. Hay et al. \cite{Hay10} show that for a general class of histogram queries, by exploiting the consistency constraints on the query output, which is differentially private by adding independent Laplace noises, one can improve the accuracy while still satisfying differential privacy. These existing works study how to efficiently answer a set of linear queries on the histogram, while our work addresses the problem of releasing the histogram itself, which can be viewed as the worst-case query release (without knowing which linear queres will be asked).

Hardt and Talwar \cite{geometry} study the tradeoff between privacy and error for answering a set of linear queries over a histogram under $\epsilon$-differental privacy. The  error is defined as the worst  expectation of the $\ell^2$-norm of the noise. \cite{geometry} derives a lower bound for the error in the high privacy regime by using tools from convex geometry and Markov's inequality, and gives an upper bound by analyzing a differentially private mechanism, $K$-norm mechanism, which is an instantiation of the exponential mechanism and involves randomly sampling from a high dimensional convex body.  The lower bound given in \cite{geometry} depends on the volume of a convex body associated with the lineary query functions, and the lower bound  works for arbitrary linear query functions. In our problem setting, the linear query functions we are studying are the histogram function, which is a speical case of   \cite{geometry}  by setting $d = n$ and setting $F$ to be the identity map function. In this case, the lower bound given in \cite{geometry} is $\Omega(\frac{\sqrt{d}}{\epsilon})$\footnote{Note that for the $d$-dimensional $\ell^1$ unit ball, the volume is $\frac{2^d}{d!}$, and thus in Theorem 3.4 of \cite{geometry}, Vol$(K)^{1/d}= \Theta(\frac{1}{d})$.}, which matches our result, as we show that for $\epsilon$-differental privacy, in the high privacy regime, adding independent Laplacian noises to each component of the histogram is asymptotically optimal in the context of $\ell^1$ and $\ell^2$ cost functions.

Nikolov, Talwar and Zhang \cite{NTZ12} extend the result of \cite{geometry} on answering linear querys over a histogram to the case of  $(\e,\delta)$-differential privacy. Using tools from discrepancy theory, convex geometry and statistical estimation, they derive lower bounds and upper bounds of the error, which are within a multiplicative factor of $O(\log \frac{1}{\delta})$ in terms of $\delta$. Their bounds work for any set of linear query functions over a histogram, while in our work we study only the identity function, i.e., the query output is the histogram itself. Our result shows that in the high privacy regime (as $(\epsilon,\delta) \to (0,0)$), the optimal error scales as $\Theta(\min(\frac{1}{\epsilon},\frac{1}{\delta}))$ and $\Theta(\min(\frac{1}{\epsilon^2},\frac{1}{\delta^2}))$ for $\ell^1$ and $\ell^2$ cost functions, respectively. Therefore, our results significantly improve the bounds in \cite{NTZ12} in terms of $\epsilon$ and $\delta$ in the high privacy regime where both $\epsilon$ and $\delta$ go to zero.

Kasiviswanathan, Rudelson, Smith and Ullman \cite{KRSU10} derive lower bounds on the noise for releasing contingency tables under $(\e,\delta)$-differential privacy constraint, where the lower bounds depend on the size and structure of the database. 
Our lower bounds are tighter and sharper than those of \cite{KRSU10} in terms of $\epsilon$ and $\delta$. For instance, in \cite{KRSU10}, for $(\e,\delta)$-differential privacy the lower bounds are proportional to $(1-\frac{\delta}{\epsilon})$, which are zero whenever $\delta = \epsilon$, while our results show that the lower bound is $\Theta(\min (\frac{1}{\epsilon}, \frac{1}{\delta}))$ as $(\epsilon,\delta) \to (0,0)$.

Anindya De \cite{De12}  studies lower  bound on the additive noise for Lipschitz query functions in $(\e,\delta)$-differential privacy which uses a different metric for the noise, and the lower bound depends on the size of the database. Jain, Kothari, and Thakurta \cite{Jain12} study how to preserve $(\epsilon,\delta)$-differential privacy for online learning algorithms, and show that the approximate differential privacy can be achieved by adding Gaussian noise to each component of the query output. They derive lower bounds on the noise, and the lower bounds can be viewed as an application of the composition theorem in \cite{DRV10} by Dwork, Rothblum, and Vadhan, which has been improved by Oh and Viswanath \cite{Oh13} recently. The difference of \cite{Jain12} and other related works from our work is that the global sensitivity of the query function is defined using $\ell^{\infty}$ metric in \cite{Jain12} and \cite{DRV10}, while in our work we use $\ell^1$ metric.

\subsection{Organization}

This paper is organized as follows. We formulate the utility-maximization/cost-minimization under the $(\e,\delta)$-differential privacy constraint for a single integer-valued query function  as a linear programming problem in Section \ref{sec:formulation}. In Section \ref{sec:zerodelta}, we study $(0,\delta)$-differential privacy, and show the nearly optimality of the simple uniform noise mechanism. In Section \ref{sec:epsilondelta}, we study the optimal mechanisms in $(\e,\delta)$-differential privacy, and show the optimality of uniform noise	mechanism and Laplacian mechanism in the regime $(\e,\delta) \to (0,0)$ in the context of $\ell^1$ and $\ell^2$ cost functions. In Section \ref{sec:zerodeltamulti}, we extend the results to the multidimensional setting for histogram-like query functions, where the query output is a vector of integers.

\section{Problem Formulation} \label{sec:formulation}


Consider an integer-valued query function
\begin{align}
	q: \database \rightarrow \Z,
\end{align}
where $\database$ is the domain of the databases.

The sensitivity of the query function $q$ is defined as
\begin{align}
	\D \triangleq \max_{ D_1,D_2 \subseteq \database: |D_1 - D_2| \le 1} | q(D_1) - q(D_2)|, \label{def:sensitivity}
\end{align}
where the maximum is taken over all possible pairs of neighboring database entries $D_1$ and $D_2$ which differ in at most one element, i.e., one is a proper subset of the other and the larger database contains just one additional element \cite{DPsurvey}. Clearly, $\D$ is an integer in this discrete setting.

\begin{definition}[$(\e,\delta)$-differential privacy \cite{DKMMN06}]
	A randomized mechanism $\KM$ gives $\e$-differential privacy if for all data sets $D_1$ and $D_2$ differing on at most one element, and all $S \subset \text{Range}(\KM)$,
	\begin{align}
	 	\text{Pr}[\KM(D_1) \in S] \le \exp(\e) \;  \text{Pr}[\KM(D_2) \in S] + \delta. \label{eqn:dpgeneral}
	 \end{align}
\end{definition}


\subsection{Operational Meaning of $(\e,\delta)$-differential privacy in the Context of Hypothesis Testing}

As shown by \cite{Zhou08}, one can interpret the   differential privacy constraint  \eqref{eqn:dpgeneral} in the context of hypothesis testing in terms of false alarm probability and missing detection probability. Indeed, consider a binary hypothesis testing problem over two neighboring datasets, $H_0: D_1 $ versus $H_1: D_2$, where an individual's record is in $D_2$ only. Given a decision rule, let $S$ be the decision region such that when the released output lies in $S$, $H_1$ will be rejected, and when the released output lies in $S^C$ (the complement of $S$), $H_0$ will be rejected. The false alarm probability $P_{FA}$ and the missing detection probability $P_{MD}$ can be written as
\begin{align}
	P_{FA} &= P(K(D_1) \in S^C), \\
	P_{MD} &= P(K(D_2) \in S).
\end{align}

Therefore, from \eqref{eqn:dpgeneral} we get
\begin{align}
	1 - P_{FA} \le e^{\e} P_{MD} + \delta.
\end{align}
Thus
\begin{align}
	 e^{\e} P_{MD} + P_{FA} \ge 1 - \delta.
\end{align}

Switch $D_1$ and $D_2$ in \eqref{eqn:dpgeneral}, and we get 
	\begin{align}
	 	\text{Pr}[\KM(D_2) \in S] \le \exp(\e) \;  \text{Pr}[\KM(D_1) \in S] + \delta. 
	 \end{align}

Therefore, 
\begin{align}
	1 - P_{MD} \le e^{\e} P_{FA} + \delta,
\end{align}
and thus
\begin{align}
	 P_{MD} + e^{\e} P_{FA} \ge 1 - \delta.
\end{align}

In conclusion,  we have  
\begin{align}
  e^{\e} P_{MD} + P_{FA} &\ge 1 - \delta, \\
	P_{MD} + e^{\e} P_{FA} &\ge 1 - \delta.
\end{align}

The $(\e,\delta)$-differential privacy constraint implies that in the context of hypothesis testing,  $P_{FA}$ and $P_{MD}$ can not be both too small.

We plot the regions of $P_{FA}$ and $P_{MD}$ under $(\e,\delta)$-differential privacy, and under two special cases: $(\e,0)$ and $(0,\delta)$-differential privacy, in Figure \ref{fig:comparison}.


\begin{figure}[h]

\begin{subfigure}[b]{0.5\linewidth}
\centering
\includegraphics[width=1.1\textwidth]{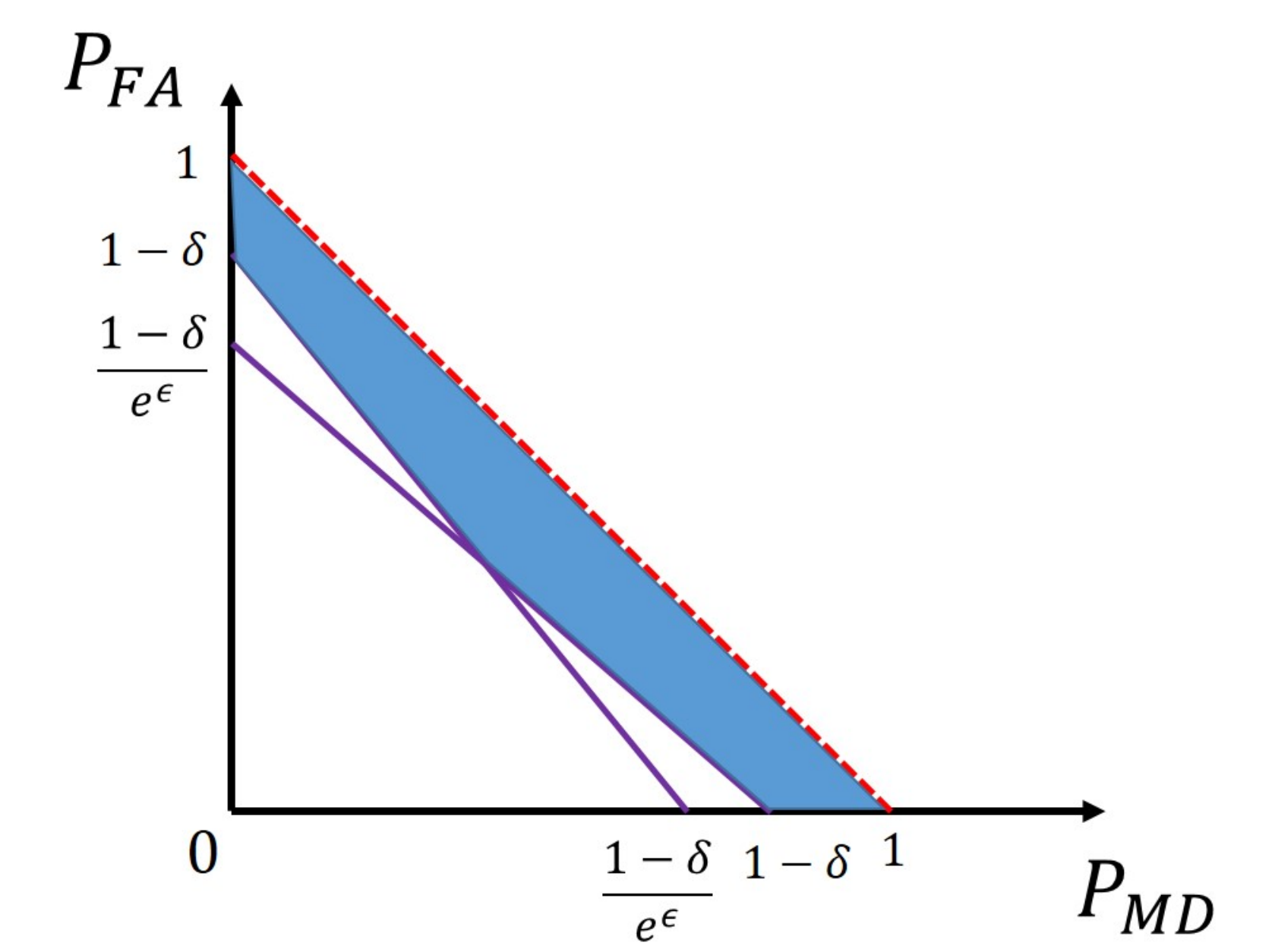}
\caption{$(\e,\delta)$-Differential Privacy}
\label{fig:compare0}
\end{subfigure}

\begin{subfigure}[b]{0.5\linewidth}
\centering
\includegraphics[width=1.1\textwidth]{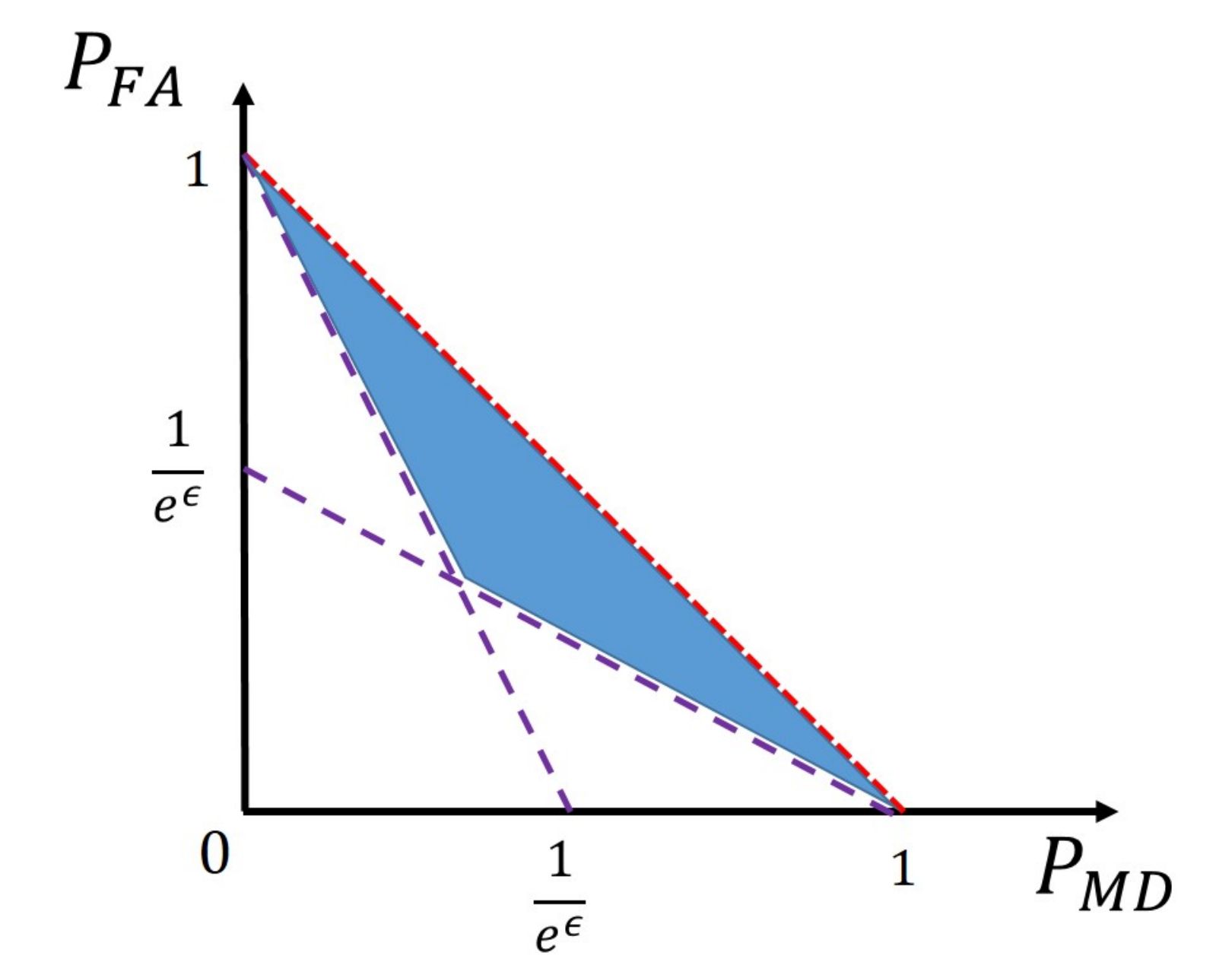}
\caption{$(\e,0)$-Differential Privacy}
\label{fig:compare1}
\end{subfigure}
\begin{subfigure}[b]{0.5\linewidth}
\centering
\includegraphics[width=1.1\textwidth]{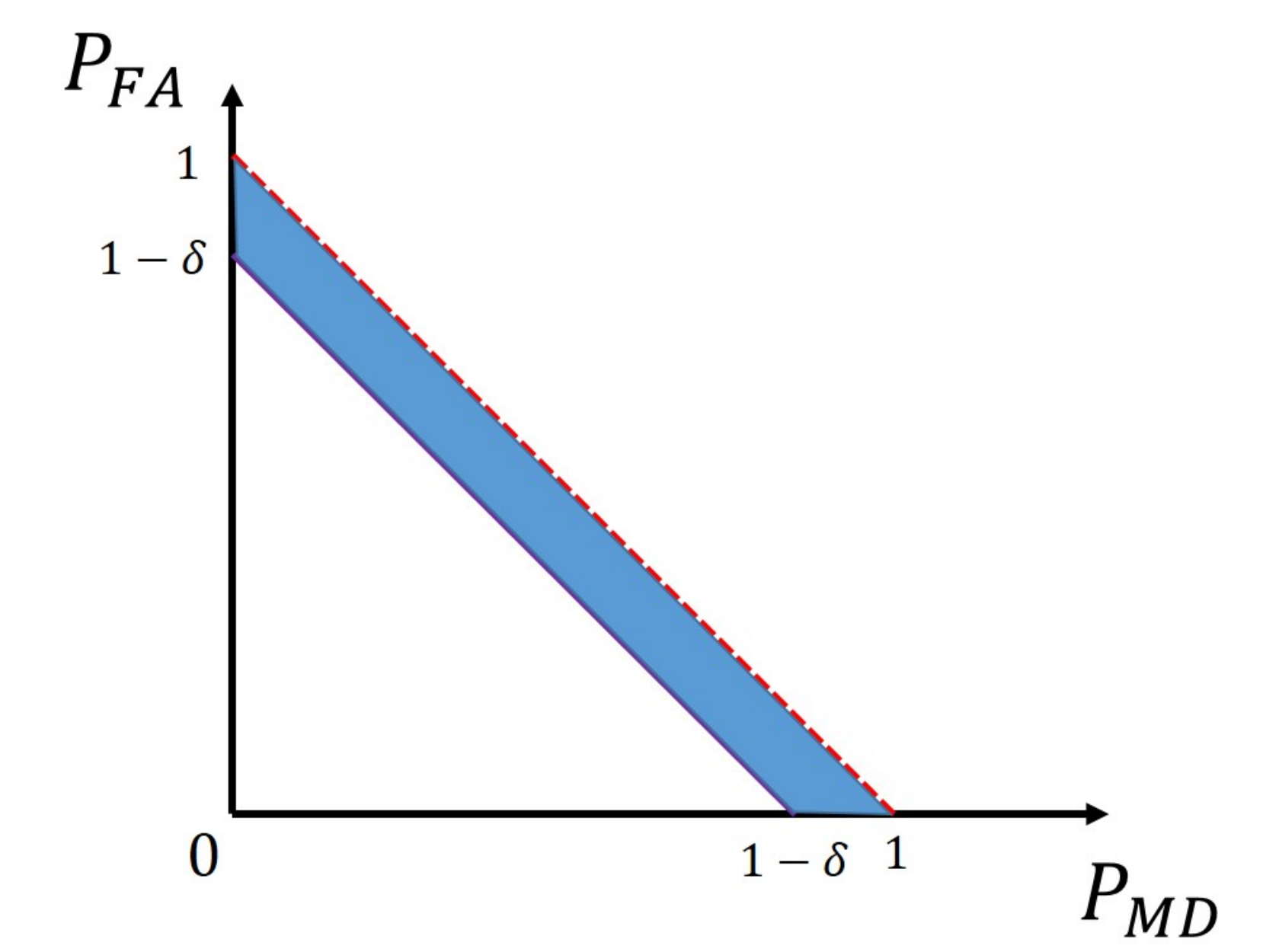}
\caption{ $(0,\delta)$-Differential Privacy}
\label{fig:compare2}
\end{subfigure}
\caption{Regions of $P_{MD}$ and $P_{FA}$ in $(\e,\delta)$, $(\e,0)$ and $(0,\delta)$-Differential Privacy. }
\label{fig:comparison}
\end{figure}


\subsection{Cost-Minimization/Utility-Maximization Formulation}
The standard approach to preserving the differential privacy is to add noise to the output of query function. Let $q(D)$ be the value of the query function evaluated at $D \subseteq \database$, the noise-adding  mechanism $\KM$ will output
\begin{align}
 	\KM(D) = q(D) + \noise,
 \end{align}
where $\noise$ is the noise added by the mechanism to the output of query function.
To make the output of the mechanism be valid, i.e., $q(D) + \noise \in \Z$, $\noise$ can only take integer values.

Let $\p$ be the probability mass function of the noise $\noise$, and use $\p_i $ to denote $\text{Pr}[\noise = i]$. For a set $S \subset \Z$, denote $\text{Pr}[\noise \in S]$ by $\p_S$.

In the following we derive the differential privacy constraint  on the probability distribution of $\noise$ from \eqref{eqn:dpgeneral}.
\begin{alignat}{2}
	 & \; \text{Pr}[\KM(D_1) \in S] & &\le  \exp(\e)\; \text{Pr}[\KM(D_2) \in S] + \delta \\
	 \Leftrightarrow & \; \text{Pr}[ q(D_1) + \noise \in S]  & &\le \exp(\e) \; \text{Pr}[q(D_2) + \noise \in S] + + \delta \\
	 \Leftrightarrow & \; \p_{S - q(D_1)}  & &\le \exp(\e) \;  \p_{S - q(D_2) } + \delta \\
	 \Leftrightarrow & \; \p_{ S'} & &\le \exp(\e) \; \p_{ S' + q(D_1)- q(D_2) } + \delta ,  \label{eqn:dpconstraint2}
\end{alignat}
where $S' \triangleq S - q(D_1) = \{ s-q(D_1) | s \in S \}$.


Since \eqref{eqn:dpgeneral} holds for any set $S \subseteq \Z$, and $|q(D_1)- q(D_2)| \le \D$, from \eqref{eqn:dpconstraint2} we have
\begin{align}
	\p_S \le \exp(\e) \; \p_{S + d} + \delta, \label{eqn:dpdiscrete}
\end{align}
for any set $S \subseteq \Z$ and for all $|d| \le \D$.


Consider a cost function $\loss(\cdot): \Z \rightarrow \R $, which is a function of the added noise $\noise$. Our goal is to minimize the expectation of the cost subject to the $(\e,\delta)$-differential privacy constraint \eqref{eqn:dpdiscrete}:
\begin{align}
	V^* := \mathop{\text{min}}\limits_{ \p} & \ \sum_{i = -\infty}^{+\infty} \loss(i) \p(i) \\
	\text{subject to} & \; \p_S \le \exp(\e) \; \p_{S + d} + \delta, \forall S \subset \Z, d \in \Z, |d| \le |\D|.  \nonumber
\end{align}

In this work, we restrict our attention to the scenario  when the cost function $\loss(k)$ is symmetric (around $k=0$) and monotonically increasing for $k \ge 0$. Furthermore, without loss of generality, we assume $\loss(0) = 0$. Using the same argument in Lemma 28  in \cite{GV13}, we only need to consider symmetric noise probability distributions.

\section{$(0,\delta)$-Differential Privacy} \label{sec:zerodelta}

We first consider the simple case when $\e = 0$, i.e., $(0,\delta)$-differential privacy. The $(0,\delta)$-differential privacy constraint requires that the total variation of the conditional probability distributions of the query output for neighboring datasets should be bounded by $\delta$.

In the differential privacy constraint \eqref{eqn:dpdiscrete}, by choosing the subset $S = S_k :=  \{ \ell : \ell \geq k \} $ for $ k \in \N$ and $d = \D$, we see that the noise probability distribution $\p$ must satisfy the constraints
 \begin{align}
 	\sum_{\ell =0}^{\D-1} \p_{k+\ell}   \leq \dlt, \quad \forall k \in \N. \label{eqn:ddpdiscrete2}
\end{align}

\subsection{$\D = 1$}

In the special case $\D = 1$, the constraints in \eqref{eqn:ddpdiscrete2} are particularly simple:
\begin{align}
p_k \leq \dlt; \quad  \forall k \geq 0.
\end{align}
For symmetric cost functions $\loss(k)$ that are monotonically increasing in $k \geq 0$, we can now
readily argue that the {\em uniform} probability distribution is optimal.

To avoid integer rounding issues, assume $\frac{1}{2\delta}$ is an integer.
\begin{theorem}
  If $\D = 1$, then
  \begin{align}
    V^* = \sum_{k=-\frac{1}{2\delta}}^{\frac{1}{2\delta}-1} \delta \loss(k),
  \end{align}
  and the optimal noise probability distribution is
  \begin{align}
    \p_k = \begin{cases}
     \delta &  -\frac{1}{2\delta} \le k \le \frac{1}{2\delta}-1 \\
  0   & \mbox{otherwise}
  \end{cases}\label{eqn:deffdeltais1}
  \end{align}
\end{theorem}

\subsection{General Lower Bound for $\D \ge 2$}

We now turn to understanding (near) optimal $(0,\dlt)$ privacy mechanisms in terms of minimizing the expected loss when the sensitivity $\D \ge 2$.


Recall that in $(0,\delta)$-differential privacy, the minimum cost  $V^*$ is the result of the following optimization problem, which is a linear program:

\begin{align}
V^* :=  \min &\quad  \sum_{k=-\infty}^{+\infty} \loss(k) \p_k \nonumber \\
\mbox{such that} &\quad   p_k \geq 0 \quad \forall k \in N \nonumber \\
 & \quad \sum_{k=-\infty}^{+\infty}\p_k = 1 \nonumber  \\
  & \quad \p_S \le   \; \p_{S + d} + \delta, \forall S \subset \Z, d \in \Z, |d| \le |\D| \label{eqn:dpconstraint6}.
\end{align}

Since $\loss(\cdot)$ is a symmetric function, we can assume $\p$ is a symmetric probability distribution. In addition, we relax the constraint \eqref{eqn:dpconstraint6} by choosing $d = \D$ and $S = S_k$ for $ k \in \N$. Then we get a relaxed linear program, the solution of which is a lower bound for $V^*$. More precisely,

\begin{align}
V_{LB} :=  \min &\quad  2\sum_{k=1}^\infty \loss(k) \p_k \label{eqn:primal}   \\
\mbox{such that} &\quad   \p_k \geq 0 \quad \forall k \in N  \nonumber \\
 &\quad \frac{\p_0}{2} + \sum_{k=1}^\infty \p_k \geq \frac{1}{2}  \label{eqn:primal111}   \\
   &\quad -\sum_{\ell=0}^{\D -1} \p_{k + \ell} \geq - \dlt, \quad \forall k \in \N  \label{eqn:primal112}.
\end{align}


To avoid integer rounding issues, assume $\frac{1}{2\delta}$ is a positive integer.


\begin{theorem}\label{thm:lowerbound_zerodelta}
If
\begin{align}
   \loss(1+\frac{\D}{2\delta})   \ge  2\left(\loss(1) + \sum_{i=1}^{\frac{1}{2\delta}} (\loss(1+i\D)-\loss(i\D)) \right), \label{eqn:LBcondition1}
 \end{align}
 then
  \begin{align}
      V^* \ge V_{LB} =  2\delta \sum_{i=0}^{\frac{1}{2\delta}-1} \loss(1 + i\D). \label{eqn:lowerbound_zerodelta}
  \end{align}
\end{theorem}

\begin{IEEEproof}
  See Appendix \ref{sec:app1}.
\end{IEEEproof}


\subsection{Uniform Noise Mechanism}

Consider the noise with the {\em uniform} probability distribution:
\begin{align}
\p_k = \begin{cases}
     \frac{\delta}{\D} &  \forall -\frac{\D}{2\delta} \le k \le \frac{\D}{2\delta}-1 \\
  0   & \mbox{otherwise}
  \end{cases} \label{eqn:uniformp}
\end{align}

It is readily verified that this noise probability distribution satisfies the $(0,\dlt)$ differential privacy constraint. Therefore, an upper bound for $V^*$ is

\begin{theorem}
   \begin{align}
     V^* \le V_{UB} \triangleq 2\sum_{i=1}^{\frac{\D}{2\delta}-1} \frac{\delta}{\D}\loss(i) + \frac{\delta}{\D}\loss(\frac{\D}{2\delta}). \label{eqn:UBzerodelta}
   \end{align}
\end{theorem}

\subsection{Comparison of $V_{LB}$ and $V_{UB}$}

We first apply the lower bound \eqref{eqn:lowerbound_zerodelta} and upper bound \eqref{eqn:UBzerodelta} to $\ell^1$ and $\ell^2$ cost functions, i.e., $\loss(i) = |i|$ and $\loss(i) = i^2$, in which $V^*$ corresponds to the minimum expected noise amplitude and minimum noise power, respectively.

Note that in the case $\loss(i) = |i|$, the condition \eqref{eqn:LBcondition1} in Theorem \ref{thm:lowerbound_zerodelta} is
\begin{align}
    \frac{\D}{2\delta} \ge \frac{1}{\delta} + 1.
\end{align}
When $\D \ge 3$, \eqref{eqn:LBcondition1} holds.

\begin{corollary} \label{cor:l1cost0delta}
  For the cost function $\loss(i) = |i|$,
  \begin{align}
     V_{LB} &= \frac{\D}{4\delta} +  1 - \frac{\D}{2}, \\
     V_{UB} &= \frac{\D}{4\delta},
  \end{align}
  and thus the additive gap
  \begin{align}
    V_{UB} - V_{LB} = \frac{\D}{2} - 1
  \end{align}
  is a constant independent of $\delta$.
\end{corollary}


In the case $\loss(i) = i^2$, the condition \eqref{eqn:LBcondition1} in Theorem \ref{thm:lowerbound_zerodelta} is
\begin{align}
    \frac{\D}{2\delta^2}(\frac{\D}{2}-1) \ge \frac{1}{\delta} + 1. \label{eqn:tttemp1}
\end{align}
When $\D \ge 3$, \eqref{eqn:tttemp1} holds.

\begin{corollary} \label{cor:l2cost0delta}
  For the cost function $\loss(i) = i^2$,
  \begin{align}
     V_{LB} &=  \frac{\D^2}{12\delta^2} - \frac{\D^2}{4\delta}   + \D(\frac{1}{2\delta}-1)    + \frac{\D^2}{6} + 1 , \\
     V_{UB} &= \frac{\D^2}{12\delta^2} + \frac{1}{6},
  \end{align}
  and thus the multiplicative gap
  \begin{align}
    \lim_{\delta \to 0 }\frac{V_{UB}}{V_{LB}} = 1.
  \end{align}

\end{corollary}


\begin{IEEEproof}
  See Appendix \ref{sec:app2}.
\end{IEEEproof}



\begin{corollary} \label{cor:polynomial}
Given a positive integer $m$, consider the cost function $\loss(i) = |i|^m$. Then
\begin{align}
   \lim_{\delta \to 0} \frac{V_{UB}}{V_{LB}} = 1.
 \end{align}

\end{corollary}

\begin{proof}
  By induction, it is easy to show that $\sum_{i=1}^n i^m = \Theta(\frac{n^{m+1}}{m+1})$, and
  \begin{align}
    \lim_{n \to +\infty} \frac{\sum_{i=1}^n i^m}{\frac{n^{m+1}}{m+1}} = 1.
  \end{align}

  Therefore,
  \begin{align}
    \lim_{\delta \to 0} \frac{V_{UB}}{V_{LB}} &= \lim_{\delta \to 0} \frac{ 2\frac{\delta}{\D}\sum_{i=1}^{\frac{\D}{2\delta}-1} i^m + \frac{\delta}{\D} \frac{\D^m}{(2\delta)^m}}{ 2\delta \sum_{i=0}^{\frac{1}{2\delta}-1} (1+i\D)^m} \\
    &= \lim_{\delta \to 0} \frac{  2\frac{\delta}{\D} \frac{\frac{\D^{m+1}}{(2\delta)^{m+1}}}{m+1} }  {   2\delta \D^m \frac{ (\frac{1}{2\delta})^{m+1}  }{m+1}  } \\
    &= 1.
  \end{align}

\end{proof}

For general cost functions, we have the following bound on the multiplicative gap between the lower bound and upper bound.

\begin{corollary}\label{cor:generalzerodelta}
  Given a cost function $\loss(\cdot)$ satisfying
  \begin{align}
     \sup_{k \ge T} \frac{\loss(k)}{\loss(k - \D + 1)} \le C,
  \end{align}
for some integer $T\in \N$, and some positive number $C \in \R$,
then
\begin{align}
  \lim_{\delta \to 0} \frac{V_{UB}}{V_{LB}} \le  1 + (1+\frac{1}{2\D})C.
\end{align}
\end{corollary}

\begin{proof}
  See Appendix \ref{sec:app3}.
\end{proof}

 \section{$(\e,\delta)$-Differential Privacy} \label{sec:epsilondelta}
 
Recall that since $\loss(\cdot)$ is a symmetric function, without loss of generality, we can  restrict ourselves to symmetric noise probability distributions, i.e.,
\begin{align}
\p_k = \p_{-k}, \forall k \in \Z.    \label{eqn:symmtry}
\end{align}

The differential privacy constraint in \eqref{eqn:dpdiscrete} can be understood in some detail by choosing the subset $S = S_k := \{ \ell : \ell \geq k \} $ for $ k \in \N$. In this case we see that the noise probability distribution  must satisfy the following constraints. For $k=0$ and $d =\D$,
 \begin{align}
  \p_{S_0} &\le e^{\e} \p_{S_{\D}} + \dlt. \label{eqn:zeroconstraint}
\end{align}

By using the symmetry condition in \eqref{eqn:symmtry} and the fact that $\sum_{\ell=-\infty}^{+\infty} \p_\ell = 1$, from \eqref{eqn:zeroconstraint} we get
\begin{align}
  \p_0  \frac{1+e^{\e}}{2}  + e^{\e} \sum_{\ell = 1}^{\D-1} \p_\ell \leq \dlt + \frac{e^\e - 1}{2}.
\end{align}

For $k = 1$ and $d = \D$, we  have
\begin{align}
\p_{S_1} \leq e^{\e} \p_{S_{\D+1}} + \dlt ,
\end{align}
and thus
\begin{align}
  \p_0  \frac{e^{\e}-1}{2}  +  e^{\e} \sum_{\ell = 1}^{ \D } \p_{\ell} \le  \dlt + \frac{e^{\e}-1}{2} .
\end{align}

For general $k \geq 2$ and $d = \D$, we  have
\begin{align}
\p_{S_k} \leq e^{\e} \p_{S_{\D+k}} + \dlt ,
\end{align}
and thus
\begin{align}
  \p_0  \frac{e^{\e}-1}{2}  +  ( e^{\e}-1 ) \sum_{\ell=1}^{k-1} \p_\ell + e^{\e}
\sum_{\ell = k}^{k+\D-1} \p_{\ell} \le  \dlt + \frac{e^{\e}-1}{2} .
\end{align}

\subsection{Lower Bound}

By restricting the set $S$ in \eqref{eqn:dpdiscrete} to be $S_k := \{ \ell : \ell \geq k \}$ for $k \in \Z$ and restricting $d$ to be $\D$, we get the following relaxed linear program, the solution of which is a lower bound for $V^*$:

\begin{align}
V_{LB} :=  \min &\quad  2\sum_{k=1}^\infty \loss(k) \p_k  \nonumber \\
\mbox{such that} &\quad   \p_k \geq 0 \quad \forall k \in N \nonumber \\
 &\quad \frac{\p_0}{2} + \sum_{k=1}^\infty \p_k \geq \frac{1}{2} \label{eqn:primalaa1}\\
 &\quad \p_0 \frac{1+\e^{\e}}{2} + e^{\e} \sum_{k=1}^{\D-1} \p_k \le \delta + \frac{e^{\e}-1}{2} \label{eqn:primalaa2} \\
 &\quad \p_0 \frac{e^{\e}-1}{2} + e^{\e} \sum_{k=1}^{\D} \p_{k} \le \delta + \frac{e^{\e}-1}{2} \label{eqn:primalaa3}  \\
 &\quad \p_0 \frac{e^{\e}-1}{2} + (e^{\e}-1) \sum_{k=1}^{i-1} \p_k + e^{\e} \sum_{k=i}^{i+\D-1} \p_k \le \delta + \frac{e^{\e}-1}{2}, \forall i \ge 2. \label{eqn:primalaa4}
\end{align}

 Define
\begin{align}
    a &\triangleq \frac{\delta + \frac{e^{\e}-1}{2}}{e^\e},\\
    b &\triangleq e^{-\e}.
\end{align}
To avoid integer rounding issues, assume that there exists an integer $n$ such that
\begin{align}
  \sum_{k=0}^{n-1} a b^k = \frac{1}{2}.
\end{align}

\begin{theorem}\label{thm:generaledloss}
If
\begin{align}
  \sum_{i=1}^{n-1} e^{-i\e} (2 \loss(i\D) - \loss(1+(i-1)\D) - \loss(1 + i\D) ) \ge \loss(1),
\end{align}
then we have
  \begin{align}
    V^* \ge V_{LB} = 2 \sum_{k=0}^{n-1} \frac{\delta + \frac{e^{\e}-1}{2}}{e^\e} e^{-k\e} \loss(1 + k\D). \label{eqn:LBed}
  \end{align}
\end{theorem}

\begin{proof}
  See Appendix \ref{sec:app4}.
\end{proof}

\subsection{ Upper Bound: Uniform Noise Mechanism and Discrete Laplacian Mechanism}

Since $(0,\delta)$-differential privacy implies $(\epsilon,\delta)$-differential privacy, we can use the uniform noise mechanism with noise probability distribution defined in \eqref{eqn:uniformp} to preserve $(\epsilon,\delta)$-differential privacy, and the corresponding upper bound is

\begin{theorem}
  For $(\epsilon,\delta)$-differential privacy, we have
  \begin{align}
    V^* \le V_{UB}^{\mbox{uniform}} = 2\sum_{i=1}^{\frac{\D}{2\delta}-1} \frac{\delta}{\D}\loss(i) + \frac{\delta}{\D}\loss(\frac{\D}{2\delta}). \label{eqn:UBuni}
  \end{align}
\end{theorem}

On the other hand, if we simply ignore the  parameter $\delta $ (i.e., set $\delta = 0$), we can use a discrete variant of Laplacian distribution to satisfy the $(\e,0)$-differential privacy, which implies $(\e,\delta)$-differential privacy.

More precisely, define $\lambda \triangleq  e^{-\frac{\e}{\D}}$.
\begin{theorem}
  The probability distribution $\p$ with
  \begin{align}
    p_k \triangleq \frac{1-\lambda}{1+\lambda} \lambda^{|k|}, \forall k \in \Z,
  \end{align}
  satisfies the $(\epsilon, \delta)$-differential privacy constraint, and the corresonding cost is
  \begin{align}
     \sum_{k=-\infty}^{+\infty} p_k \loss(k) = 2 \sum_{k=1}^{+\infty} \frac{1-\lambda}{1+\lambda} \lambda^k \loss(k).
  \end{align}
\end{theorem}

\begin{corollary}
  \begin{align}
    V^* \le  V_{UB}^{\mbox{Lap}} \triangleq   2 \sum_{k=1}^{+\infty}  \frac{1-\lambda}{1+\lambda} \lambda^k \loss(k) \label{eqn:UBlap}
  \end{align}
\end{corollary}

\subsection{Comparison of Lower Bound and Upper Bound}

In this section, we compare the lower bound \eqref{eqn:LBed} and the upper bounds $V_{UB}^{\mbox{uniform}} $ and $V_{UB}^{\mbox{Lap}}$ for $(\e,\delta)$-differential privacy for $\ell^1$ and $\ell^2$ cost functions, i.e., $\loss(i) = |i|$ and $\loss(i) = i^2$, in which $V^*$ corresponds to the minimum expected noise amplitude and minimum noise power, respectively. We show that the multiplicative gap between the lower bound and upper bound is bounded by a constant as $(\e,\delta) \to (0,0)$.

\subsubsection{ $\e \le \delta$ regime}

We first compare the gap between the lower bound $V_{LB}$ and the upper bound $V_{UB}^{\mbox{uniform}} $ in the regime $\e \le \delta$ as $\delta \to 0$.

\begin{corollary}\label{cor:l1edelta}
  For the cost function  $\loss(k) = |k|$, in the regime $\epsilon \le \delta$, we have
  \begin{align}
    \lim_{\delta \to 0} \frac{V_{UB}^{\mbox{uniform}}}{V_{LB}} \le \frac{1}{4(1 - 2 \log \frac{3}{2})} \approx 1.32
  \end{align}
\end{corollary}

\begin{proof}
  See Appendix \ref{sec:app5}.
\end{proof}


\begin{corollary}\label{cor:l2edelta}
  For the cost function  $\loss(k) = k^2$, in the regime $\epsilon \le \delta$, we have
  \begin{align}
      \lim_{\delta \to 0} \frac{V_{UB}^{\mbox{uniform}}}{V_{LB}} \le  \frac{1}{12 (2 - 4 \log (\frac{3}{2})  -   2 (\log (\frac{3}{2}))^2 )} \approx \frac{5}{3}.
  \end{align}
\end{corollary}

\begin{proof}
  See Appendix \ref{sec:app6}.
\end{proof}


\subsubsection{ $\delta \le \epsilon$ regime}

We then compare the gap between the lower bound $V_{LB}$ and the upper bound $V_{UB}^{\mbox{Lap}} $ in the regime $\delta \le \e$ as $\e \to 0$.

\begin{corollary}\label{cor:l1deltae}
  For the cost function  $\loss(k) = |k|$, in the regime $\delta \le \e$, we have
  \begin{align}
    \lim_{\e \to 0} \frac{V_{UB}^{\mbox{Lap}}}{V_{LB}} \le \frac{1}{1 - 2 \log \frac{3}{2}} \approx 5.29.
  \end{align}
\end{corollary}

\begin{proof}
  See Appendix \ref{sec:app7}.
\end{proof}


\begin{corollary}\label{cor:l2deltae}
  For the cost function  $\loss(k) = k^2$, in the regime $\epsilon \le \delta$, we have
  \begin{align}
      \lim_{\delta \to 0} \frac{V_{UB}^{\mbox{Lap}}}{V_{LB}} \le  \frac{2}{(2 - 4 \log (\frac{3}{2})  -   2 (\log (\frac{3}{2}))^2 )} \approx 40.
  \end{align}
\end{corollary}

\begin{proof}
  See Appendix \ref{sec:app8}.
\end{proof}

\section{$(\e,\delta)$-Differential Privacy in the Multi-dimensional Setting} \label{sec:zerodeltamulti}

In this section we consider the $(\e,\delta)$-differential privacy in the multi-dimensional setting, where the query output has multiple components and the global sensitivity $\D$ is defined as the maximum $\ell^1$ norm of the difference of the query outputs over two neighboring datasets.

Let $d$ be the dimension of the query output. Hence, the query output $q(D) \in \Z^d$.
Let $\p$ be the probability mass function of the additive noise over the domain $\Z^d$. Then the $(\e,\delta)$-differential privacy constraint on $\p$ in the multi-dimensional setting is that
\begin{align}
  \p_S \le \p_{S+\mathbf{v}} + \delta, \forall S \subset \Z^d, \mathbf{v} \in \Z^d,  \|\mathbf{v}\|_1 \le \D. \label{eqn:multidp}
\end{align}

Consider a cost function $\loss(\cdot): \Z^d \rightarrow \R $, which is a function of the added noise $\noise$. Our goal is to minimize the expectation of the cost subject to the $(\e,\delta)$-differential privacy constraint \eqref{eqn:multidp}:
\begin{align}
  V^* := \mathop{\text{min}}\limits_{ \p} & \ \sum_{\mathbf{v} \in \Z^d}  \loss(\mathbf{v}) \p(\mathbf{v}) \label{eqn:optdis}\\
  \text{subject to} & \;   \p_S \le \p_{S+\mathbf{v}} + \delta, \forall S \subset \Z^d, \mathbf{v} \in \Z^d,  \|\mathbf{v}\|_1 \le \D.  \nonumber
\end{align}

\subsection{$(0,\delta)$-differential privacy}

We first consider the simple case when $\e = 0$, i.e., $(0,\delta)$-differential privacy. The $(0,\delta)$-differential privacy constraint requires that the total variation of the conditional probability distributions of the query output for neighboring datasets should be bounded by $\delta$.

In the differential privacy constraint \eqref{eqn:multidp}, by choosing the subset
\begin{align}
  S = S_k^m :=  \{ (i_1,i_2,\dots,i_d) \in \Z^d \; | \; i_m \geq k \}
\end{align}
for $ k \in \N$, $m \in \{1,2,\dots,d\}$, and choosing $\mathbf{v}$ such that only one compoment is $\D$ and all other components are zero, we see that the noise probability distribution $\p$ must satisfy the constraints
 \begin{align}
  \sum_{ (i_1,i_2,\dots,i_d) \in \Z^d: k \le i_m \le k+\D-1  }  \p(i_1,i_2,\dots,i_d)   \leq \delta, \quad \forall k \in \N, \forall m \in \{1,2,\dots,d\}.
\end{align}

To avoid integer-rounding issues, we assume that $\frac{1}{2\delta}$ is an integer.

\subsubsection{Lower Bound on $V^*$}

We relax the  constraint \eqref{eqn:multidp}  by choosing $S$ to be $S_k^m$ and choosing $\mathbf{v}$ such that only one compoment is $\D$ and all other components are zero. Then we get a relaxed linear program, the solution of which is a lower bound for $V^*$. More precisely,

\begin{align}
V^* \ge V_{LB} :=  \min &\quad  \sum_{ \mathbf{i} \in \Z^d} \p(\mathbf{i}) \loss(\mathbf{i})  \label{eqn:primalmulti} \\
\mbox{such that} &\quad   \p(\mathbf{i}) \geq 0 \quad \forall \mathbf{i} \in \Z^d \nonumber \\
&\quad \sum_{ \mathbf{i} \in \Z^d} \p(\mathbf{i}) \geq 1  \nonumber  \\
&\quad  \sum_{ (i_1,i_2,\dots,i_d) \in \Z^d: k \le i_m \le k+\D-1  }  \p(i_1,i_2,\dots,i_d)   \leq \delta , \quad \forall k \in \N, \forall m \in \{1,2,\dots,d\} . \nonumber
\end{align}

\begin{theorem}\label{thm:l1multilbzerodelta}
  In the case $\loss(\mathbf{i}) = \|\mathbf{i}\|_1, \forall \mathbf{i} \in \Z^d$, we have
  \begin{align}
    V_{LB} \ge \frac{d\D}{4\delta} - \frac{\D-1}{2}d.
  \end{align}
\end{theorem}

\begin{proof}
  See Appendix \ref{sec:app9}.
\end{proof}

\begin{theorem}\label{thm:l2multilbzerodelta}
  In the case $\loss(\mathbf{i}) = \|\mathbf{i}\|_2^2 = \sum_{m=1}^d i_m^2, \forall \mathbf{i} = (i_1,\dots,i_d)\in \Z^d$, we have
  \begin{align}
    V_{LB} \ge \frac{d\D^2}{12\delta^2}  + (\frac{1}{\D}-1) \frac{d\D^2}{4\delta} + \frac{1-\D}{2}d + \frac{d\D^2}{6}.
  \end{align}
\end{theorem}

\begin{proof}
See Appendix \ref{sec:app11}.
\end{proof}

\subsubsection{Uniform Noise Mechanism in the Multi-Dimensional Setting}

Consider the noise with the {\em uniform} probability distribution:
\begin{align}
\p(i_1,i_2,\dots,i_d) = \begin{cases}
     \frac{\delta^d}{\D^d} &     -\frac{\D}{2\delta} \le i_m \le \frac{\D}{2\delta}-1, \forall m \in \{1,2,\dots,d\} \\
  0   & \mbox{otherwise}
  \end{cases}  \label{eqn:uniformmulti}
\end{align}

It is readily verified that this noise probability distribution satisfies the $(0,\dlt)$ differential privacy constraint \eqref{eqn:multidp}. Therefore, an upper bound for $V^*$ is

\begin{theorem}
   \begin{align}
     V^* \le V_{UB} \triangleq \sum_{(i_1,i_2,\dots,i_d)\in \Z^d \; | \; -\frac{\D}{2\delta} \le i_m \le \frac{\D}{2\delta}-1 , \forall m \in \{1,2,\dots,d\} } \frac{\delta^d}{\D^d} \loss(i_1,i_2,\dots,i_d)  . \label{eqn:UBzerodeltamulti}
   \end{align}
\end{theorem}

\begin{corollary}\label{cor:uniformhigh}
  In the case $\loss(\mathbf{i}) = \|\mathbf{i}\|_1, \forall \mathbf{i} \in \Z^d$, we have
  \begin{align}
    V_{UB} = \frac{d\D}{4\delta}.
  \end{align}
\end{corollary}

\begin{IEEEproof}
\begin{align}
  V_{UB} &= \sum_{(i_1,i_2,\dots,i_d)\in \Z^d \; | \; -\frac{\D}{2\delta} \le i_m \le \frac{\D}{2\delta}-1 , \forall m \in \{1,2,\dots,d\} } \frac{\delta^d}{\D^d} \loss(i_1,i_2,\dots,i_d) \\
  &= \sum_{i_1 = -\frac{\D}{2\delta} }^{\frac{\D}{2\delta}-1 } \cdots \sum_{i_d = -\frac{\D}{2\delta} }^{\frac{\D}{2\delta}-1 } \frac{\delta^d}{\D^d} (|i_1| + \cdots + |i_d|) \\
  &= d \sum_{i_1 = -\frac{\D}{2\delta} }^{\frac{\D}{2\delta}-1 } \cdots \sum_{i_d = -\frac{\D}{2\delta} }^{\frac{\D}{2\delta}-1 } \frac{\delta^d}{\D^d} |i_1| \\
  &= d \left(\frac{\D}{\delta}  \right)^{d-1} \sum_{i_1 = -\frac{\D}{2\delta} }^{\frac{\D}{2\delta}-1 } \frac{\delta^d}{\D^d} |i_1| \\
  &= d \left(\frac{\D}{\delta}  \right)^{d-1} \frac{\delta^d}{\D^d} \left( \frac{(1+\frac{\D}{2\delta}) \frac{\D}{2\delta}}{2}   + \frac{\frac{\D}{2\delta}(\frac{\D}{2\delta}-1) }{2} \right) \\
  &= \frac{d\D}{4\delta}.
\end{align}

\end{IEEEproof}

 \begin{corollary}\label{cor:uniformhighl2}
  In the case $\loss(\mathbf{i}) = \|\mathbf{i}\|_2^2 \triangleq \sum_{m=1}^d i_m^2, \forall \mathbf{i} = (i_1,\dots,i_d) \in \Z^d$, we have
  \begin{align}
    V_{UB} = \frac{d\D^2}{12\delta^2} + \frac{d}{6}.
  \end{align}
\end{corollary}

\begin{IEEEproof}
\begin{align}
  V_{UB} &= \sum_{(i_1,i_2,\dots,i_d)\in \Z^d \; | \; -\frac{\D}{2\delta} \le i_m \le \frac{\D}{2\delta}-1 , \forall m \in \{1,2,\dots,d\} } \frac{\delta^d}{\D^d} \loss(i_1,i_2,\dots,i_d) \\
  &= \sum_{i_1 = -\frac{\D}{2\delta} }^{\frac{\D}{2\delta}-1 } \cdots \sum_{i_d = -\frac{\D}{2\delta} }^{\frac{\D}{2\delta}-1 } \frac{\delta^d}{\D^d} (|i_1|^2 + \cdots + |i_d|^2) \\
  &= d \sum_{i_1 = -\frac{\D}{2\delta} }^{\frac{\D}{2\delta}-1 } \cdots \sum_{i_d = -\frac{\D}{2\delta} }^{\frac{\D}{2\delta}-1 } \frac{\delta^d}{\D^d} |i_1|^2 \\
  &= d \left(\frac{\D}{\delta}  \right)^{d-1} \sum_{i_1 = -\frac{\D}{2\delta} }^{\frac{\D}{2\delta}-1 } \frac{\delta^d}{\D^d} |i_1|^2 \\
  &= d \left(\frac{\D}{\delta}  \right)^{d-1} \frac{\delta^d}{\D^d} \left(   \frac{\frac{\D}{2\delta} (1+\frac{\D}{2\delta}) (\frac{\D}{\delta}+1) }{6}   + \frac{ (\frac{\D}{2\delta}-1)\frac{\D}{2\delta} (\frac{\D}{\delta}-1) }{6} \right) \\
  &= \frac{d\D^2}{12\delta^2} + \frac{d}{6}.
\end{align}

\end{IEEEproof}

\subsubsection{Comparison of Lower Bound and Upper Bound for $\ell^1$  Cost Function }

\begin{corollary}
   For the cost function $\loss(\mathbf{i}) = \|\mathbf{i}\|_1$,
  \begin{align}
     V_{LB} &\ge \frac{d\D}{4\delta} - \frac{\D-1}{2}d, \\
     V_{UB} &= \frac{d\D}{4\delta},
  \end{align}
  and thus the additive gap
  \begin{align}
    V_{UB} - V_{LB} \le \frac{\D-1}{2}d,
  \end{align}
  which is a constant independent of $\delta$.

\end{corollary}

In the case that $\D = 1$, the additive gap $\frac{\D-1}{2}d$ is zero, and thus $V_{LB} = V_{UB}$.

\begin{corollary}
  For the cost function $\loss(\mathbf{i}) = \|\mathbf{i}\|_1$, if $\D = 1$, then
  \begin{align}
     V^* = V_{UB} = V_{LB} = \frac{d\D}{4\delta},
  \end{align}
  and thus the uniform noise mechanism is optimal in this setting.
\end{corollary}


\begin{corollary}
   For the cost function $\loss(\mathbf{i}) = \|\mathbf{i}\|_2^2$,
  \begin{align}
     V_{LB} &\ge \frac{d\D^2}{12\delta^2}  + (\frac{1}{\D}-1) \frac{d\D^2}{4\delta} + \frac{1-\D}{2}d + \frac{d\D^2}{6}, \\
     V_{UB} &= \frac{d\D^2}{12\delta^2} + \frac{d}{6},
  \end{align}
  and thus
  \begin{align}
     \lim_{\delta \to 0} \frac{V_{UB}}{V_{LB}} = 1.
  \end{align}

\end{corollary}

In the case that $\D = 1$,
\begin{align}
  V_{LB} \ge  \frac{d}{12\delta^2} +  \frac{d}{6} = V_{UB},
\end{align}
and thus $V_{LB} = V_{UB}$.

\begin{corollary}
  For the cost function $\loss(\mathbf{i}) = \|\mathbf{i}\|_2^2$, if $\D = 1$, then
  \begin{align}
     V^* = V_{UB} = V_{LB} = \frac{d}{12\delta^2} +  \frac{d}{6},
  \end{align}
  and thus the uniform noise mechanism is optimal in this setting.
\end{corollary}


\subsection{$(\epsilon,\delta)$-differential privacy}

The $(\epsilon,\delta)$-differential privacy constraint on the probability mass function $\p$ in the multi-dimensional setting is that
\begin{align}
  \p_S \le e^\e \p_{S+\mathbf{v}} + \delta, \forall S \subset \Z^d, \mathbf{v} \in \Z^d,  \|\mathbf{v}\|_1 \le \D.
\end{align}

We relax this  constraint   by choosing $S$ to be $S_k^m$ and choosing $\mathbf{v}$ such that only one component is $\D$ and all other components are zero. Then we get a relaxed linear program, the solution of which is a lower bound for $V^*$. More precisely,

\begin{align}
V^* \ge V_{LB} :=  \min &\quad  \sum_{ \mathbf{i} \in \Z^d} \p(\mathbf{i}) \loss(\mathbf{i})  \label{eqn:primalmultied} \\
\mbox{such that} &\quad   \p(\mathbf{i}) \geq 0 \quad \forall \mathbf{i} \in \Z^d \nonumber \\
&\quad \sum_{ \mathbf{i} \in \Z^d} \p(\mathbf{i}) \geq 1     \nonumber \\
&\quad \forall k \in \N, \forall m \in \{1,2,\dots,d\}, \nonumber \\
  \sum_{ (i_1,i_2,\dots,i_d) \in \Z^d: k \le i_m \le k+\D-1  } & \p(i_1,i_2,\dots,i_d) - (e^{\e} -1) \sum_{ (i_1,i_2,\dots,i_d) \in \Z^d:   i_m \ge k+\D  }  \p(i_1,i_2,\dots,i_d)    \leq \delta .  \nonumber
\end{align}

We are interested in characterizing $V^*$ for the $\ell^1$ and $\ell^2$ cost functions in the high privacy regime when $(\e,\delta) \to (0,0)$.

\subsubsection{Lower Bound for $\ell^1$ Cost Function}

The dual linear program of \eqref{eqn:primalmultied} for $\ell^1$ cost function $\loss(\mathbf{i}) = \|\mathbf{i} \|_1$ is that

\begin{align}
   V_{LB} :=  \max \quad  \mu - \delta \left (\sum_{i_1 \in \Z} y_{i_1}^{(1)}  + \sum_{i_2 \in \Z} y_{i_2}^{(2)} + \cdots + \sum_{i_d \in \Z} y_{i_d}^{(d)} \right) \\
 \\
\mbox{such that}\quad  y_{i_1}^{(1)}, y_{i_2}^{(2)}, \dots, y_{i_d}^{(d)} \ge 0, \forall i_1 \in \Z, i_2 \in \Z, \dots, i_d \in \Z  \\
  \mu - \sum_{ i_1 \in [k_1 -\D +1, k_1]} y_{i_1}^{(1)}  + (e^\e -1)\sum_{i_1 \le k_1 - \D} y_{i_1}^{(1)}  \nonumber \\
  - \dots  - \sum_{ i_d \in [k_d -\D +1, k_d]} y_{i_d}^{(d)}  + (e^\e -1)\sum_{i_d \le k_d - \D} y_{i_d}^{(d)}   \nonumber \\
  \le |k_1| + |k_2| + \cdots + |k_d|, \forall (k_1,\dots,k_d)\in \Z^d  .
\end{align}

Given the parameters $(\e,\delta)$, let $\beta = \max(\e,\delta)$. Since $(\beta,\beta)$-differential privacy is a relaxed version of $(\e,\delta)$-differential privacy, in the above dual program we can replace both $\e$ and $\delta$ by $\beta$, and the optimal value of the objecitve function will still be a lower bound of $V^*$. More precisely,

\begin{align}
   V^* \ge V'_{LB} :=  \max \quad  \mu - \beta \left (\sum_{i_1 \in \Z} y_{i_1}^{(1)}  + \sum_{i_2 \in \Z} y_{i_2}^{(2)} + \cdots + \sum_{i_d \in \Z} y_{i_d}^{(d)} \right) \\
 \\
\mbox{such that}\quad  y_{i_1}^{(1)}, y_{i_2}^{(2)}, \dots, y_{i_d}^{(d)} \ge 0, \forall i_1 \in \Z, i_2 \in \Z, \dots, i_d \in \Z  \\
  \mu - \sum_{ i_1 \in [k_1 -\D +1, k_1]} y_{i_1}^{(1)}  + (e^\beta -1)\sum_{i_1 \le k_1 - \D} y_{i_1}^{(1)}  \nonumber \\
  - \dots  - \sum_{ i_d \in [k_d -\D +1, k_d]} y_{i_d}^{(d)}  + (e^\beta -1)\sum_{i_d \le k_d - \D} y_{i_d}^{(d)}   \nonumber \\
  \le |k_1| + |k_2| + \cdots + |k_d|, \forall (k_1,\dots,k_d)\in \Z^d  .
\end{align}

\begin{theorem}\label{eqn:loweredelta}
For the $\ell^1$ cost function,
  \begin{align}
  \lim_{\max(\e,\delta) \to 0 } \frac{V'_{LB}}{\frac{d\D}{\max(\e,\delta)}} \ge  \log \frac{9}{8} \approx 0.1178
\end{align}
\end{theorem}

\begin{proof}
  See Appendix \ref{sec:app10}.
\end{proof}

Similarly, for the $\ell^2$ cost function, we have the lower bound

\begin{align}
   V^* \ge V'_{LB} :=  \max \quad  \mu - \beta \left (\sum_{i_1 \in \Z} y_{i_1}^{(1)}  + \sum_{i_2 \in \Z} y_{i_2}^{(2)} + \cdots + \sum_{i_d \in \Z} y_{i_d}^{(d)} \right) \\
 \\
\mbox{such that}\quad  y_{i_1}^{(1)}, y_{i_2}^{(2)}, \dots, y_{i_d}^{(d)} \ge 0, \forall i_1 \in \Z, i_2 \in \Z, \dots, i_d \in \Z  \\
  \mu - \sum_{ i_1 \in [k_1 -\D +1, k_1]} y_{i_1}^{(1)}  + (e^\beta -1)\sum_{i_1 \le k_1 - \D} y_{i_1}^{(1)}  \nonumber \\
  - \dots  - \sum_{ i_d \in [k_d -\D +1, k_d]} y_{i_d}^{(d)}  + (e^\beta -1)\sum_{i_d \le k_d - \D} y_{i_d}^{(d)}   \nonumber \\
  \le |k_1|^2 + |k_2|^2 + \cdots + |k_d|^2, \forall (k_1,\dots,k_d)\in \Z^d  .
\end{align}

\begin{theorem}\label{eqn:loweredeltal2}
For the $\ell^2$ cost function,
  \begin{align}
   \lim_{\max(\e,\delta) \to 0 } \frac{V'_{LB}}{\frac{d\D^2}{\beta^2}} \ge  0.0177.
\end{align}
\end{theorem}

\begin{proof}
  See Appendix \ref{sec:app12}.
\end{proof}

\subsubsection{Upper Bounds: Uniform Noise Mechanism and Discrete Laplacian Mechanism}

Since $(0,\delta)$-differential privacy implies  $(\epsilon,\delta)$-differential privacy and we have shown that the uniform noise mechanism defined in \eqref{eqn:uniformmulti} satisfies $(0,\delta)$-differential privacy, an upper bound for $V^*$ for the $\ell^1$ cost function is
\begin{align}
   V^* \le V_{UB}^{\mbox{uniform}} = \frac{d\D}{4\delta} \label{eqn:upp11}
\end{align}
by Corollary \ref{cor:uniformhigh}.

In addition,  $(\e,0)$-differential privacy also implies  $(\e,\delta)$-differential privacy, and the discrete Laplacian mechanism satisfies $(\e,0)$-differential privacy. Consider the discrete Laplacian mechanism in the multi-dimensional setting with probability mass function $\p$ defined as
\begin{align}
  \p(i_1,i_2,\dots,i_d) = \left( \frac{1 - \lambda}{ 1 + \lambda} \right)^{d} \lambda^{|i_1| + |i_2| + \cdots + |i_d|}, \forall (i_1,\dots,i_d) \in \Z^d,
\end{align}
where $\lambda \triangleq e^{-\frac{\e}{\D}}$.

The corresponding cost achieved by Laplacian mechanism for the $\ell^1$ cost function is
\begin{align}
  V_{UB}^{\mbox{Lap}} &= \sum_{(i_1,i_2,\dots,i_d) \in \Z^d} \left( \frac{1 - \lambda}{ 1 + \lambda} \right)^{d} \lambda^{|i_1| + |i_2| + \cdots + |i_d|} (|i_1| + |i_2| + \cdots + |i_d|) \\
  &= \frac{2d\lambda}{1 - \lambda^2} \\
  &= \frac{2d e^{-\frac{\e}{\D}}}{1 - e^{-2\frac{\e}{\D}}} \\
  &= \Theta (\frac{d\D}{\e}), \label{eqn:upp22}
\end{align}
as $\e \to 0$.

Similarly, for the $\ell^2$ cost function, we have
\begin{align}
    V_{UB}^{\mbox{uniform}} &= \frac{d\D^2}{12\delta^2} + \frac{d}{6},
\end{align}
and
\begin{align}
    V_{UB}^{\mbox{Lap}}  &= \sum_{(i_1,i_2,\dots,i_d) \in \Z^d} \left( \frac{1 - \lambda}{ 1 + \lambda} \right)^{d} \lambda^{|i_1| + |i_2| + \cdots + |i_d|} (|i_1|^2 + |i_2|^2 + \cdots + |i_d|^2)  \\
    &= \frac{2d\lambda}{(1-\lambda)^2} \\
    &= \Theta( \frac{2d\D^2}{\e^2} ).
\end{align}

\subsubsection{Comparison of Lower Bound and Upper Bounds}

Compare the lower bound  in Theorem \ref{eqn:loweredelta} and the upper bounds \eqref{eqn:upp11} and \eqref{eqn:upp22}, and we conclude that for the $\ell^1$ cost function, the multiplicative gap between the upper bound and lower bound is upper bounded by a constant as $(\e,\delta) \to (0,0)$. More precisely,

\begin{corollary}
  For the $\ell^1$ cost function, we have
  \begin{align}
    V'_{LB}  \le  V^*  \le  \min(V_{UB}^{\mbox{uniform}} ,V_{UB}^{\mbox{Lap}}),
  \end{align}
  and as $(\e,\delta) \to (0,0)$,
  \begin{align}
      \lim_{(\e,\delta) \to (0,0)} \frac{ \min(V_{UB}^{\mbox{uniform}} ,V_{UB}^{\mbox{Lap}})}{V'_{LB}} \le \frac{1}{\log \frac{9}{8}} \approx 8.49
    \end{align}
\end{corollary}

Similarly, for the $\ell^2$ cost function, we have
\begin{corollary}
  For the $\ell^2$ cost function, we have
  \begin{align}
    V'_{LB}  \le  V^*  \le  \min(V_{UB}^{\mbox{uniform}} ,V_{UB}^{\mbox{Lap}}),
  \end{align}
  and as $(\e,\delta) \to (0,0)$,
  \begin{align}
      \lim_{(\e,\delta) \to (0,0)} \frac{ \min(V_{UB}^{\mbox{uniform}} ,V_{UB}^{\mbox{Lap}})}{V'_{LB}} \le \frac{2}{0.0177} \approx 113.
    \end{align}
\end{corollary}


\section{Acknowledgement} The authors thank Kamalika Chaudhury for helpful discussions. 

\appendices

\section{Proof of Theorem \ref{thm:lowerbound_zerodelta} } \label{sec:app1}

\begin{proof}[Proof of Theorem \ref{thm:lowerbound_zerodelta} ]
  Consider a feasible solution to the optimization problem \eqref{eqn:primal} with primal variables
  \begin{align}
       p_k &=
       \begin{cases}
     \delta &   k = 1 + i \D, \text{for} \quad i = 0,1,2,\dots,\frac{1}{2\delta}-1 \\
  0   & \mbox{otherwise}
  \end{cases}
\end{align}
The corresponding value of the objective function is
\begin{align}
  2\delta \sum_{i=0}^{\frac{1}{2\delta}-1} \loss(1+i\D).
\end{align}
Therefore,
\begin{align}
  V_{LB} \le 2\delta \sum_{i=0}^{\frac{1}{2\delta}-1} \loss(1+i\D). \label{eqn:primalachieve}
\end{align}
  We claim that the above primal variables are the optimal solution. We prove this claim by constructing the corresponding dual variables.

Associating dual variables  $\mu$ with the constraint in \eqref{eqn:primal111}, $y_k$ with the constraint in \eqref{eqn:primal112}, we have the dual linear program:

\begin{align}
V_{LB} = \max &  \quad  \mu - 2\dlt \sum_{k=0}^\infty y_k  \nonumber \\
\mbox{such that} &\quad  \mu \geq 0, y_k \geq 0, \forall k \in \N,  \label{eq:dual1}\\
&\quad \frac{1}{2} \mu - y_0 \le 0, \\
&\quad \mu - \sum_{i= \max(0,k-\D+1)}^k y_k  \leq  \loss(k), \forall k \ge 1. \label{eq:dual2}
\end{align}

The complementary slackness conditions require that
\begin{align}
  \mu - y_0 - y_1 &= \loss(1), \\
  \mu - \sum_{i=2+(k-1)\D}^{1+k\D} y_k &= \loss(1+k\D), \text{for} \; k=1,2,\dots,\frac{1}{2\delta}-1,\\
  y_k &= 0 , \forall k \ge (\frac{1}{2\delta} - 1)\D + 2.
\end{align}

Consider the following dual variables:
\begin{align}
  \mu &= \loss(1 + \frac{\D}{2\delta}),\\
  y_k &= 0, \forall k \ge (\frac{1}{2\delta} - 1)\D + 2,\\
  y_k &= \loss(k + \D) - \loss(k+\D-1) + y(k +\D), \forall  2 \le k \le (\frac{1}{2\delta} - 1)\D + 1,\\
  y_1 &= \sum_{i=1}^{\frac{1}{2\delta}} (\loss(1+i\D) - \loss(i\D)) \ge 0, \\
  y_0 &= \mu - \loss(1) - y_1 = \loss(1+\frac{\D}{2\delta}) - \loss(1) - \sum_{i=1}^{\frac{1}{2\delta}} (\loss(1+i\D) - \loss(i\D)) \ge 0.
\end{align}

It is easy to verify that these dual variables satisfy the constraints of the dual linear program, and the value of the objective function is
\begin{align}
     \mu - 2\delta \sum_{k=0}^{+\infty} y_k = & \mu - 2\delta  \sum_{i=0}^{\frac{1}{2\delta}-1} (\mu - \loss(1 + i\D)) \\
   =& 2\delta \sum_{i=0}^{\frac{1}{2\delta}-1} \loss(1+i\D).
\end{align}

Therefore, by weak duality we have
\begin{align}
  V_{LB} \ge 2\delta \sum_{i=0}^{\frac{1}{2\delta}-1} \loss(1+i\D).
\end{align}

Due to \eqref{eqn:primalachieve}, we conclude
\begin{align}
  V_{LB} = 2\delta \sum_{i=0}^{\frac{1}{2\delta}-1} \loss(1+i\D).
\end{align}

\end{proof}


\section{Proof of Corollary \ref{cor:l2cost0delta}} \label{sec:app2}

\begin{proof}[Proof of Corollary \ref{cor:l2cost0delta}]
First we compute the lower bound $V_{LB}$ via
  \begin{align}
    V_{LB} &= 2 \sum_{i=0}^{\frac{1}{2\delta}-1} \delta \loss(1+i\D) \\
    &= 2 \delta \sum_{i=0}^{\frac{1}{2\delta}-1} (1+i\D)^2 \\
    &= 2 \delta \sum_{i=0}^{\frac{1}{2\delta}-1} (1+2i\D + i^2 \D^2) \\
    &= 2 \delta (\frac{1}{2\delta} + 2\D \frac{\frac{1}{2\delta} (\frac{1}{2\delta}-1)}{2} + \D^2 \frac{ (\frac{1}{2\delta}-1) \frac{1}{2\delta} (2\frac{1}{2\delta}-1) }{6}) \\
    &= 1 + \D(\frac{1}{2\delta}-1) + \frac{\D^2}{12\delta^2} + \frac{\D^2}{6} - \frac{\D^2}{4\delta} \\
    &= \Theta(\frac{\D^2}{12\delta^2}).
  \end{align}

The upper bound is
\begin{align}
  V_{UB} &= 2 \sum_{i=1}^{\frac{\D}{2\delta}-1} \frac{\delta}{\D}\loss(i) + \frac{\delta}{\D}\loss(\frac{\D}{2\delta}) \\
  &= 2 \frac{\delta}{\D} \frac{ (\frac{\D}{2\delta}-1) \frac{\D}{2\delta} (\frac{\D}{\delta}-1) }{6} + \frac{\delta}{\D}\frac{\D^2}{4\delta^2} \\
  &= \frac{1}{6}(\frac{\D^2}{2\delta^2}+1-\frac{3\D}{2\delta}) + \frac{\D}{4\delta} \\
  &= \frac{\D^2}{12\delta^2} + \frac{1}{6} \\
  &= \Theta(\frac{\D^2}{12\delta^2}).
\end{align}

Therefore, the multiplicative gap goes to one as $\delta \to 0$, i.e.,
\begin{align}
  \lim_{\delta \to 0} \frac{V_{UB}}{V_{LB}} = 1.
\end{align}

\end{proof}


\section{Proof of Corollary \ref{cor:generalzerodelta}} \label{sec:app3}

\begin{proof}[Proof of Corollary \ref{cor:generalzerodelta}]
  Using the fact that $\loss(\cdot)$ is a monotonically increasing function for $k \ge 0$, we have
  \begin{align}
    V_{UB} - V_{LB} &= 2\sum_{i=1}^{\frac{\D}{2\delta}-1} \frac{\delta}{\D}\loss(i) + \frac{\delta}{\D}\loss(\frac{\D}{2\delta})   - 2\delta \sum_{i=0}^{\frac{1}{2\delta}-1} \loss(1+i\D)  \\
      &\le -2\delta \loss(1) + \frac{\delta}{\D}\loss(\frac{\D}{2\delta}) + 2\delta \loss(\frac{\D}{2\delta}-1) \\
      &\le (2+\frac{1}{\D}) \delta \loss(\frac{\D}{2\delta}).
  \end{align}

Therefore,
\begin{align}
    \frac{V_{UB}}{V_{LB}} &=  1 + \frac{V_{UB}-V_{LB}}{V_{LB}} \\
      &\le 1 + \frac{ (2+\frac{1}{\D}) \delta \loss(\frac{\D}{2\delta})  }{2\delta \sum_{i=0}^{\frac{1}{2\delta}-1} \loss(1+i\D)} \\
      &\le 1 + \frac{ (2+\frac{1}{\D}) \delta \loss(\frac{\D}{2\delta})  }{2\delta   \loss(1+(\frac{1}{2\delta}-1)\D)},
\end{align}
and thus
\begin{align}
   \lim_{\delta \to 0}  \frac{V_{UB}}{V_{LB}} &\le  1 + (1 + \frac{1}{2\D})C.
\end{align}

\end{proof}


\section{Proof of Theorem \ref{thm:generaledloss}} \label{sec:app4}

\begin{proof}[Proof of Theorem \ref{thm:generaledloss}]
  Consider the feasible primal variables $\{p_k\}_{k \in \N}$ defined as
  \begin{align}
    \p_k &=   \begin{cases}
     a b^i &   \text{for} \; k = 1 + i\D, 0 \le i \le n-1 \\
  0   & \mbox{otherwise}
  \end{cases}\label{eqn:defprimal1}
 \end{align}
It is straightforward to verify that the above primal variables satisfy the constraints of the relaxed linear program, and the corresponding value of the  objective function is
\begin{align}
   2 \sum_{k=0}^{n-1} a b^{k} \loss(1 + k\D).
\end{align}

We prove it is also the optimal value by constructing the optimal dual variables for the corresponding dual linear program.

Associating dual variables $\mu, y_0, y_1, y_i$ with the primal constraints in \eqref{eqn:primalaa1},\eqref{eqn:primalaa2},\eqref{eqn:primalaa3} and
\eqref{eqn:primalaa4}, respectively, we have the dual linear program:

\begin{align}
V_{LB} :=  \min &\quad  \mu - (2\delta + e^{\e} - 1) \sum_{k=0}^{+\infty} y_k \\
\mbox{such that} &\quad  \mu \ge 0, y_k \geq 0 \quad \forall k \in N \label{eqn:dualll2}  \\
 &\quad \frac{1}{2} \mu - \frac{1+ e^{\e}}{2} y_0  - \frac{e^{\e}-1}{2} y_1 - \frac{e^{\e}-1}{2} \sum_{k=2}^{+\infty} y_k \le 0      \label{eqn:dualll1}  \\
 &\quad  \mu - e^{\e}y_0 - e^{\e} y_1 - (e^{\e}-1) \sum_{k=2}^{+\infty} y_k  \le \loss(1) \label{eqn:dualll3}  \\
 &\quad \mu - e^{\e} \sum_{l=\max(0,k-\D+1)}^k y_l - (e^{\e}-1)\sum_{l=k+1}^{+\infty} y_l \le \loss(k), \forall k \ge 2. \label{eqn:dualll4}
\end{align}

If the primal variables defined in \eqref{eqn:defprimal1} are the optimal solution, the complementary slackness conditions require that the corresponding dual variables satisfy that
\begin{align}
  \mu &= \loss(1) + e^{\e} (y_0+y_1) + (e^{\e}-1)\sum_{l=2}^{+\infty} y_l \\
  \mu &= \loss(1+\D) + e^{\e} \sum_{l=2}^{1+\D} y_l + (e^{\e}-1)\sum_{l = 2+\D}^{+\infty} y_l \\
  \mu &= \loss(1+k\D) + e^{\e} \sum_{l=2+(k-1)\D}^{1+k\D} y_l + (e^{\e}-1)\sum_{l=2+k\D}^{+\infty} y_l, \forall 1\le k \le n-1, \\
  y_l &= 0, \forall l \ge 2 + (n-1)\D.
\end{align}

Consider the following dual variables defined via
\begin{align}
  \mu &= \loss(1 + (n-1)\D), \\
  y_k &= 0, \forall k \ge 2 + (n-2)\D, \\
  y_k &= b (y_{k+\D} + \loss(k+\D) - \loss(k+\D-1)  ), \forall 2\le k \le  1+ (n-2)\D, \\
  y_1 &= \sum_{i=1}^{n-1} b^i (\loss(1+i\D)-\loss(i\D)), \\
  y_0 &= \sum_{i=1}^{n-1} b^i ( \loss(i\D) - \loss(1+ (i-1)\D) ).
\end{align}

We verify that the above dual variables satisfy the inequality \eqref{eqn:dualll1} in the following
\begin{align}
  & (1+e^{\e})y_0 + (e^\e -1)y_1 + (e^{\e}-1) \sum_{k=2}^{+\infty} y_k -\mu \ge 0 \\
\Leftrightarrow & y_0 - y_1 + e^\e (y_0 + y_1) + (e^{\e}-1) \sum_{k=2}^{+\infty} y_k - \mu \ge 0\\
\Leftrightarrow & y_0 - y_1 + \mu -\loss(1) - \mu \ge 0 \\
\Leftrightarrow & y_0 - y_1  -\loss(1) \ge 0 \\
\Leftrightarrow &\sum_{i=1}^{n-1} b^i (2 \loss(i\D) - \loss(1+(i-1)\D) - \loss(1 + i\D) ) \ge \loss(1).
\end{align}

It is easy to verify that the dual variables satisfy the constraints \eqref{eqn:dualll2}, \eqref{eqn:dualll1}, \eqref{eqn:dualll3} and \eqref{eqn:dualll4} in the dual linear program. Next we compute  the corresponding value of the objective function
\begin{align}
   & \mu - (2\delta + e^{\e} - 1) \sum_{k=0}^{+\infty} y_k \\
   =& \mu - (2\delta + e^{\e} - 1) ( y_0 + y_1 + \frac{\mu - \loss(1) - e^\e (y_0+y_1)}{e^\e - 1}) \\
   =& \mu - \frac{2\delta + e^{\e} - 1}{e^\e - 1} ( \mu - \loss(1) - y_0 - y_1) \\
   =& \loss(1+(n-1)\D) - \frac{2\delta + e^\e -1}{e^\e - 1} (\loss(1+(n-1)\D) - \loss(1) - \sum_{i=1}^{n-1}b^i (\loss(1+i\D) - \loss(1+(i-1)\D)) ) \\
   =& 2 \sum_{k=0}^{n-1} a b^{k} \loss(1 + k\D),
\end{align}
which is also the value of the objecitve function in the primal problem achieved by the primal variables defined in \eqref{eqn:defprimal1}. Therefore, we conclude that
\begin{align}
   V_{LB} = 2 \sum_{k=0}^{n-1} a b^{k} \loss(1 + k\D).
\end{align}

\end{proof}


\section{Proof of Corollary \ref{cor:l1edelta}} \label{sec:app5}

\begin{proof}[Proof of Corollary \ref{cor:l1edelta}]
  For the cost function $\loss(k) = |k|$,
  \begin{align}
    V_{LB} &= 2 \sum_{k=0}^{n-1} a b^{k} \loss(1 + k\D) \\
           &= 2 \sum_{k=0}^{n-1} a b^{k}  (1 + k\D) \\
           &= 1 + 2 a \D \sum_{k=0}^{n-1}  b^{k} k  \\
           &= 1 + 2 a\D ( \frac{b-b^n}{(1-b)^2}  - \frac{(n-1)b^n}{1-b}).
  \end{align}

Given $\delta >0$, $V_{LB}$ is a decreasing function of $\epsilon$. Therefore, to lower bound $\frac{V_{UB}^{\mbox{uniform}}}{V_{LB}}$ in the regime $\e \le \delta$, we only need to consider the case $\epsilon = \delta $. Thus, in the following we set $\epsilon = \delta$.

Since $\sum_{k=0}^{n-1}a b^k = \frac{1}{2}$, we have
\begin{align}
  a\frac{1-b^n}{1-b} &= \frac{1}{2} \\
  \Leftrightarrow b^n &= 1 - \frac{1-b}{2a}.
\end{align}

As $\delta \to 0$,  $\frac{1-b}{2a} = \frac{1- e^{-\e}}{2 \frac{\delta + \frac{e^\e - 1}{2}}{e^\e}} \to \frac{1}{3}$, and thus
\begin{align}
  \lim_{\delta \to 0} b^n &= 1 - \frac{1}{3} = \frac{2}{3}, \\
  n &= \Theta(\frac{\log (\frac{3}{2})}{\e}).
\end{align}

Note that $a = \Theta(\frac{3}{2} \delta) $ as $\delta \to 0$.

Therefore, as $\delta \to 0$,
\begin{align}
  V_{LB} &\approx 2 \D a (\frac{1 -  \frac{2}{3} }{\e^2} - \frac{\frac{\log (\frac{3}{2})}{\e}\frac{2}{3} }{\e}) \\
  &\approx 2 \D  \frac{3}{2} \delta ( \frac{1}{3\delta^2} - \frac{\frac{2}{3} \log (\frac{3}{2})}{\delta^2}) \\
  &= \frac{\D}{\delta} (1 - 2 \log \frac{3}{2}) \\
  &\approx 0.19 \frac{\D}{\delta}.
\end{align}

Recall $V_{UB}^{\mbox{uniform}} = \frac{\D}{4\delta}$.

Therefore,
\begin{align}
  \lim_{\e = \delta \to 0} \frac{V_{UB}}{V_{LB}} = \frac{1}{4(1 - 2 \log \frac{3}{2})} \approx 1.32,
\end{align}

and thus
\begin{align}
  \lim_{\e \le \delta \to 0} \frac{V_{UB}}{V_{LB}} \le \frac{1}{4(1 - 2 \log \frac{3}{2})} \approx 1.32,
\end{align}

\end{proof}


\section{Proof of Corollary \ref{cor:l2edelta}} \label{sec:app6}

\begin{proof}[Proof of Corollary \ref{cor:l2edelta}]
  Using the same argument in the proof of Corollary \ref{cor:l1edelta}, we can set $\epsilon = \delta$.

  For the cost function $\loss(k) = k^2$,
  \begin{align}
    V_{LB} &= 2 \sum_{k=0}^{n-1} a b^{k} \loss(1 + k\D) \\
           &= 2 \sum_{k=0}^{n-1} a b^{k}  (1 + k\D)^2 \\
           &= 1 + 4 a \D \sum_{k=0}^{n-1}  b^{k} k  +  2 a \D^2 \sum_{k=0}^{n-1}  b^{k} k^2   \\
           &\approx 2 a \D^2 \sum_{k=0}^{n-1}  b^{k} k^2 \\
           &= 2 \D^2 \frac{\delta + \frac{e^\e -1}{2}}{e^{\e}}  \frac{-b + 2 ( \frac{b(1-b^{n-1})}{(1-b)^2} - \frac{(n-1) b^{n}}{1-b} ) - \frac{b^2(1-b^{n-2})}{1-b} - (n-1)^2 b^{n}}{1-b}\\
           &\approx 2\D^2 \frac{3}{2}\e \frac{  2(\frac{1-\frac{2}{3}}{\e^2} - \frac{\frac{2}{3}\log (\frac{3}{2})}{\e^2} )  -  \frac{1}{3\e}  - \frac{2}{3} \frac{(\log (\frac{3}{2}))^2}{\e^2} }{\e} \\
           &\approx \frac{3\D^2}{\e^2} (\frac{2}{3} - \frac{4}{3} \log (\frac{3}{2})  -   \frac{2}{3} (\log (\frac{3}{2}))^2 )\\
            &= \frac{\D^2}{\e^2} (2 - 4 \log (\frac{3}{2})  -   2 (\log (\frac{3}{2}))^2 )\\
           &\approx  \frac{\D^2}{20\e^2} \\
           &=   \frac{\D^2}{20\delta^2}
 \end{align}

Recall $V_{UB}^{\mbox{uniform}} = \frac{\D^2}{12\delta^2}$.

Therefore,
\begin{align}
  \lim_{\epsilon = \delta \to 0} \frac{V_{UB}^{\mbox{uniform}}}{V_{LB}} =  \frac{1}{12 (2 - 4 \log (\frac{3}{2})  -   2 (\log (\frac{3}{2}))^2 )} \approx \frac{5}{3},
\end{align}
and thus
\begin{align}
  \lim_{\e \le \delta \to 0} \frac{V_{UB}^{\mbox{uniform}}}{V_{LB}} \le  \frac{1}{12 (2 - 4 \log (\frac{3}{2})  -   2 (\log (\frac{3}{2}))^2 )} \approx \frac{5}{3}.
\end{align}

\end{proof}


\section{Proof of Corollary \ref{cor:l1deltae}} \label{sec:app7}

\begin{proof}[Proof of Corollary \ref{cor:l1deltae}]

  For the cost function $\loss(k) = |k|$,
  \begin{align}
    V_{LB} &= 2 \sum_{k=0}^{n-1} a b^{k} \loss(1 + k\D) \\
           &= 2 \sum_{k=0}^{n-1} a b^{k}  (1 + k\D) \\
           &= 1 + 2 a \D \sum_{k=0}^{n-1}  b^{k} k  \\
           &= 1 + 2 a\D ( \frac{b-b^n}{(1-b)^2}  - \frac{(n-1)b^n}{1-b}).
  \end{align}

Given $\e >0$, $V_{LB}$ is a decreasing function of $\delta$. Therefore, to lower bound $\frac{V_{UB}^{\mbox{Lap}}}{V_{LB}}$ in the regime $\delta \le \e$, we only need to consider the case $\delta  = \epsilon $. Thus, in the following we set $\delta = \epsilon$.

Following the same calculations in the proof of Corollary \ref{cor:l1edelta}, we have
\begin{align}
  V_{LB} &\approx \frac{\D}{\delta} (1 - 2 \log \frac{3}{2}) \\
  &\approx 0.19 \frac{\D}{\delta} \\
  &=  0.19 \frac{\D}{\e}.
\end{align}

On the other hand,
we have
\begin{align}
  V_{UB}^{\mbox{Lap}} &=  2 \sum_{k=1}^{+\infty}  \frac{1-\lambda}{1+\lambda} \lambda^k k \\
  &= \frac{2e^{-\frac{\e}{\D}}}{1 - e^{-2\frac{\e}{\D}}} \\
  &\approx \frac{\D}{\e},
\end{align}
as $\e \to 0$.

Therefore,
\begin{align}
  \lim_{\e = \delta \to 0} \frac{V_{UB}^{\mbox{Lap}}}{V_{LB}} = \frac{1}{1 - 2 \log \frac{3}{2}} \approx 5.29,
\end{align}
and thus
\begin{align}
  \lim_{\e \le \delta \to 0} \frac{V_{UB}^{\mbox{Lap}}}{V_{LB}} \le \frac{1}{1 - 2 \log \frac{3}{2}} \approx 5.29.
\end{align}

\end{proof}


\section{Proof of Corollary \ref{cor:l2deltae}} \label{sec:app8}

\begin{proof}[Proof of Corollary \ref{cor:l2deltae}]
  Using the same argument in the proof of Corollary \ref{cor:l1deltae}, we can set $\epsilon = \delta$.

  For the cost function $\loss(k) = k^2$, following the same calculations in the proof of Corollary \ref{cor:l2edelta}, we have
  \begin{align}
    V_{LB} &\approx \frac{\D^2}{\e^2} (2 - 4 \log (\frac{3}{2})  -   2 (\log (\frac{3}{2}))^2 )\\
           &\approx  \frac{\D^2}{20\e^2}
 \end{align}

On the other hand, we have
\begin{align}
  V_{UB}^{\mbox{Lap}} &=  2 \sum_{k=1}^{+\infty}  \frac{1-\lambda}{1+\lambda} \lambda^k k^2 \\
  &= \frac{2 \lambda}{(1 - \lambda)^2} \\
  &\approx 2\frac{\D^2}{\e^2},
\end{align}
as $\e \to 0$.

Therefore,
\begin{align}
  \lim_{\epsilon = \delta \to 0} \frac{V_{UB}^{\mbox{Lap}}}{V_{LB}} =  \frac{2}{(2 - 4 \log (\frac{3}{2})  -   2 (\log (\frac{3}{2}))^2 )} \approx 40,
\end{align}
and thus
\begin{align}
  \lim_{\e \le \delta \to 0} \frac{V_{UB}^{\mbox{Lap}}}{V_{LB}} \le  \frac{2}{(2 - 4 \log (\frac{3}{2})  -   2 (\log (\frac{3}{2}))^2 )} \approx 40.
\end{align}

\end{proof}


\section{Proof of Theorem \ref{thm:l1multilbzerodelta}} \label{sec:app9}

\begin{IEEEproof}[Proof of Theorem \ref{thm:l1multilbzerodelta}]
  Consider the dual program of the linear program \eqref{eqn:primalmulti},

\begin{align*}
   V_{LB} :=  \max &\quad  \mu - \delta \left (\sum_{i_1 \in \Z} y_{i_1}^{(1)}  + \sum_{i_2 \in \Z} y_{i_2}^{(2)} + \cdots + \sum_{i_d \in \Z} y_{i_d}^{(d)} \right) \\
 \\
\mbox{such that} &\quad  y_{i_1}^{(1)}, y_{i_2}^{(2)}, \dots, y_{i_d}^{(d)} \ge 0, \forall i_1 \in \Z, i_2 \in \Z, \dots, i_d \in \Z  \\
  \mu - \sum_{ i_1 \in [k_1 -\D +1, k_1]} y_{i_1}^{(1)} - &\dots - \sum_{ i_d \in [k_d -\D +1, k_d]} y_{i_d}^{(d)} \le |k_1| + |k_2| + \cdots + |k_d|, \forall (k_1,\dots,k_d)\in \Z^d   .
\end{align*}

Consider a candidate solution with
\begin{align}
  \mu &= \frac{d\D}{2\delta}
\end{align}
and for all $m \in \{1,2,\dots,d\}$,
\begin{align}
y_i^{(m)} = \begin{cases}
     \frac{\mu}{d} &     i = 0 \\
  \max(\frac{\mu}{d} - k \D, 0)   &  i = k\D, \; \text{for}\; k \in \Z, k \ge 1 \\
  \max(\frac{\mu}{d} - (|k|-1) \D - 1, 0)   &  i = k\D, \; \text{for}\; k \in \Z, k \le -1 \\
 0  & \mbox{otherwise}
  \end{cases}
\end{align}

It is easy to verify that this candidate solution satisfies the constraints, and the corresponding  value of the objective function is
\begin{align}
  & \mu - \delta \left (\sum_{i_1 \in \Z} y_{i_1}^{(1)}  + \sum_{i_2 \in \Z} y_{i_2}^{(2)} + \cdots + \sum_{i_d \in \Z} y_{i_d}^{(d)} \right) \\
  = & \mu - \delta d  \sum_{i_1 \in \Z} y_{i_1}^{(1)} \\
  = & \mu - \delta d  \left(  \sum_{i=0}^{\frac{\mu}{d\D}} (\frac{\mu}{d} - i\D) + \sum_{i=0}^{\frac{\mu}{d\D}-1} (\frac{\mu}{d} - i\D -1)    \right ) \\
  = & \mu - \delta d \left( \frac{ \frac{\mu}{d} ( \frac{\mu}{d\D}  + 1 ) }{2}  + \frac{(\frac{\mu}{d}+\D-2)\frac{\mu}{d\D}}{2}    \right) \\
  = & \mu - \delta d (\frac{\mu^2}{d^2\D} + \frac{\mu}{d} - \frac{\mu}{d\D}) \\
  = & \mu - \delta(\frac{\mu^2}{d\D} + \mu - \frac{\mu}{\D}) \\
  = & \frac{d\D}{4\delta} - \frac{\D-1}{2}d.
\end{align}


Therefore, we have
\begin{align}
   V_{LB} \ge \frac{d\D}{4\delta} - \frac{\D-1}{2}d.
\end{align}

\end{IEEEproof}


\section{Proof of Theorem \ref{thm:l2multilbzerodelta}} \label{sec:app11}

\begin{proof}[Proof of Theorem \ref{thm:l2multilbzerodelta}]

 Consider the dual program of the linear program \eqref{eqn:primalmulti},
\begin{align*}
   V_{LB} :=  \max &\quad  \mu - \delta \left (\sum_{i_1 \in \Z} y_{i_1}^{(1)}  + \sum_{i_2 \in \Z} y_{i_2}^{(2)} + \cdots + \sum_{i_d \in \Z} y_{i_d}^{(d)} \right) \\
 \\
\mbox{such that} &\quad  y_{i_1}^{(1)}, y_{i_2}^{(2)}, \dots, y_{i_d}^{(d)} \ge 0, \forall i_1 \in \Z, i_2 \in \Z, \dots, i_d \in \Z  \\
  \mu - \sum_{ i_1 \in [k_1 -\D +1, k_1]} y_{i_1}^{(1)} - \dots - & \sum_{ i_d \in [k_d -\D +1, k_d]} y_{i_d}^{(d)} \le |k_1|^2 + |k_2|^2 + \cdots + |k_d|^2, \forall (k_1,\dots,k_d)\in \Z^d   .
\end{align*}

To avoid integer-rounding issues, assume that $\frac{1}{2\delta}$ is an integer. Consider a candidate solution with
\begin{align}
  \mu &= \frac{d\D^2}{4\delta^2}
\end{align}
and for all $m \in \{1,2,\dots,d\}$,
\begin{align}
y_i^{(m)} = \begin{cases}
     \frac{\mu}{d} &     i = 0 \\
   \frac{\mu}{d} - k^2 \D^2    &  i = k\D, \; \text{for}\; 1 \le k \ge \frac{1}{2\delta} \\
  \frac{\mu}{d} - \left( (|k|-1) \D + 1\right)^2   &  i = k\D, \; \text{for}\;  - \frac{1}{2\delta}  \le k \le -1 \\
 0  & \mbox{otherwise}
  \end{cases}
\end{align}

It is easy to verify that this candidate solution satisfies the constraints, and the corresponding  value of the objective function is
\begin{align}
  & \mu - \delta \left (\sum_{i_1 \in \Z} y_{i_1}^{(1)}  + \sum_{i_2 \in \Z} y_{i_2}^{(2)} + \cdots + \sum_{i_d \in \Z} y_{i_d}^{(d)} \right) \\
  = & \mu - \delta d  \sum_{i_1 \in \Z} y_{i_1}^{(1)} \\
  = & \mu - \delta d  \left(  \sum_{i=0}^{\frac{1}{2\delta}} (\frac{\mu}{d} - i^2\D^2) + \sum_{i=0}^{\frac{1}{2\delta}-1} (\frac{\mu}{d} - (i\D +1)^2 )    \right ) \\
  = &  \mu - \delta d \left(  (\frac{1}{2\delta}+1)\frac{\mu}{d} - \D^2 \frac{ \frac{1}{2\delta}(\frac{1}{2\delta}+1)(\frac{1}{\delta}+1) }{6}   + \frac{1}{2\delta}\frac{\mu}{d} - \frac{1}{2\delta} - \D^2 \frac{ (\frac{1}{2\delta} -1) \frac{1}{2\delta}(\frac{1}{\delta}-1) }{6}   -   \D  \frac{1}{2\delta} (\frac{1}{2\delta} - 1) \right) \\
  = & \mu - \delta d \left(   (\frac{1}{\delta}+1)\frac{\mu}{d} - \frac{\D^2 \frac{1}{2\delta} ( \frac{1}{2 \delta^2} + 1 ) }{3} - \frac{1}{2\delta} - \D \frac{1}{2\delta} (\frac{1}{2\delta} - 1)       \right) \\
  = & \frac{d\D^2}{12\delta^2}  + (\frac{1}{\D}-1) \frac{d\D^2}{4\delta} + \frac{1-\D}{2}d + \frac{d\D^2}{6}.
\end{align}


Therefore, we have
\begin{align}
   V_{LB} \ge \frac{d\D^2}{12\delta^2}  + (\frac{1}{\D}-1) \frac{d\D^2}{4\delta} + \frac{1-\D}{2}d + \frac{d\D^2}{6}.
\end{align}

\end{proof}


\section{Proof of Theorem \ref{eqn:loweredelta}} \label{sec:app10}

\begin{IEEEproof}[Proof of Theorem \ref{eqn:loweredelta}]
Consider a candidate solution with $\mu = \frac{d\D \log \frac{3}{2}}{\beta}$  (assuming $k \triangleq \frac{\mu}{d\D}$ is an integer), and
for all $m \in \{1,2,\dots,d\}$,
\begin{align}
y_i^{(m)} = \begin{cases}
     0 &     i \le -k\D \\
   e^\beta y_{i-\D}^{(m)} + 1   &  i  \in [-k\D+1, 0] \\
  \max( e^\beta y_{i-\D}^{(m)} - 1, 0)   &  i  \ge 0 .
  \end{cases}
\end{align}

It is easy to verify that the above candidate solution satisfies the constraints of the dual linear program. We can derive the analytical expression for $y_i^{m}$, which is
\begin{align}
y_i^{(m)} = \begin{cases}
     0 &     i \le -k\D \\
    \frac{ e^{ (k-j)\beta} - 1 }{e^\beta - 1}   &  i  \in [-(j+1)\D + 1, -j \D], \text{for} j \in [0, k-1] \\
  \max( e^{j\beta} \frac{e^{k\beta} - 2}{e^\beta -1} + \frac{1}{e^\beta -1} , 0)   &  i \in [ (j-1)\D+1, j\D] .
  \end{cases}
\end{align}

To avoid integer-rounding issues, assume that $n \triangleq \frac{1}{\beta} \log \frac{1}{2-e^{k\beta}} = \frac{\log 2}{\beta}$ is an integer. Then the value of the objective function with this candidate solution is
\begin{align}
  & \mu - \beta \left (\sum_{i_1 \in \Z} y_{i_1}^{(1)}  + \sum_{i_2 \in \Z} y_{i_2}^{(2)} + \cdots + \sum_{i_d \in \Z} y_{i_d}^{(d)} \right) \\
  = & \mu - \beta d  \sum_{i_1 \in \Z} y_{i_1}^{(1)} \\
  = & \mu - \beta d \D \left(  \sum_{i=1}^{k} \frac{e^{i\beta}-1}{e^\beta-1}  + \sum_{i=1}^{n}  (e^{i\beta} \frac{e^{k\beta}-2}{e^\beta -1} + \frac{1}{e^\beta-1})   \right )  \\
  = & \mu - \beta d \D \left(  \frac{ \frac{e^\beta (1 - e^{k\beta}) }{1 - e^\beta} - k }{e^\beta - 1}  + \frac{e^{k\beta}-2}{e^\beta - 1} \frac{e^{\beta}(1-e^{n\beta})}{1 - e^\beta} + \frac{n}{e^\beta - 1} \right) \\
  = &  \frac{ d\D \log \frac{3}{2} }{\beta} - \beta d \D \left(  \frac{ \frac{e^\beta (1 - \frac{3}{2}) }{1 - e^\beta} - \frac{\log \frac{3}{2}}{\beta} }{e^\beta - 1}  + \frac{-\frac{1}{2}}{e^\beta - 1} \frac{e^\beta (1 - 2)}{1 - e^{\beta}}  + \frac{\log 2}{\beta(e^\beta - 1)}    \right)  \\
  = & \frac{ d\D \log \frac{3}{2} }{\beta} - \beta d \D \left( \frac{e^\beta}{2(e^\beta-1)^2}  - \frac{\log \frac{3}{2}}{\beta(e^\beta-1)}  - \frac{e^\beta}{2(e^\beta -1)^2}  + \frac{\log 2}{\beta(e^\beta - 1)}  \right)   \\
  = & \Theta \left( \frac{d\D}{\beta} ( \log \frac{3}{2} - \frac{1}{2} + \log \frac{3}{2} + \frac{1}{2} - \log 2 )  \right)  \\
  = &  \log \frac{9}{8}\Theta \left(   \frac{d\D}{\beta}  \right) \\
  \approx & \Theta \left( 0.1178  \frac{d\D}{\beta}  \right),
\end{align}
as $\beta \triangleq \max(\e, \delta) \to 0$.

Therefore,
\begin{align}
  \lim_{\max(\e,\delta) \to 0 } \frac{V'_{LB}}{\frac{d\D}{\beta}} \ge  \log \frac{9}{8} \approx 0.1178
\end{align}

\end{IEEEproof}


\section{Proof of Theorem \ref{eqn:loweredeltal2}} \label{sec:app12}

\begin{proof}[Proof of Theorem \ref{eqn:loweredeltal2}]

Let $\alpha = \frac{3}{2}$.
Consider a candidate solution with $\mu = \frac{d\D^2 \log^2 \alpha }{\beta^2}$  (assuming $k \triangleq \frac{\sqrt{\frac{\mu}{d}}}{\D} = \frac{\log \alpha}{\beta}$ is an integer), and
for all $m \in \{1,2,\dots,d\}$,
\begin{align}
y_i^{(m)} = \begin{cases}
     0 &     i \le -k\D \\
   e^\beta y_{i-\D}^{(m)} + 2|i|+1   &  i  \in [-k\D+1, 0] \\
  \max( e^\beta y_{i-\D}^{(m)} - (2i + 1), 0)   &  i  \ge 0 .
  \end{cases}
\end{align}

It is easy to verify that the above candidate solution satisfies the constraints of the dual linear program.

Define
\begin{align}
 z_1 &= \frac{2}{e^{\beta}-1}, \\
 z_2 &= \frac{1 - \frac{2e^{\beta}\D}{e^{\beta}-1}}{e^{\beta}-1}, \\
 z_3 &= \frac{2}{1-e^{\beta}}, \\
 z_4 &= \frac{1 - \frac{2e^{\beta}\D}{1-e^{\beta}} }{1 - e^{\beta}}.
\end{align}

We can derive the analytical expression for $y_i^{m}$, which is
\begin{align}
y_i^{(m)} = \begin{cases}
     0 &     i \le -k\D \\
     e^{(k-k')\beta} \left( z_1 (k\D + j) + z_2   \right)  - z_1 (k'\D + j) - z_2  &  i = - (k'\D + j),   \text{for} k' \in [0, k-1], j \in [0, \D-1]
  \end{cases}
\end{align}

and for  $i = (m-1)\D + j$,  where $j \in [1,\D], m \ge 1$,
\begin{align}
  y_i^{(m)} =  \max(  a_{m,j}  , 0) ,
\end{align}

where
\begin{align}
  a_{m,j} \triangleq   e^{m\beta} \left(z_1 (k\D + \D - j) +  z_2  ) - z_1(\D-j) - z_2 - z_3(\D - j) + z_4   \right)  - z_4 - z_3 ( (m-1)\D + j ) .
\end{align}

For each $j \in [1,\D]$, and we are interested in finding the number $m(j)$ such that $a_{m(j),j} = 0$. As $\beta \to 0$, from $a_{m(j),j} = 0$, we get
\begin{align}
  e^{m(j)\beta} e^{k\beta} (\frac{2}{\beta}k\D - \frac{2\D}{\beta^2}) = -\frac{2\D}{\beta^2} - \frac{2}{\beta}m(j)\D + o(\frac{1}{\beta^2}).
\end{align}

Therefore,
\begin{align}
  m(j) = \frac{\log \gamma}{\beta} + o(\frac{1}{\beta}),
\end{align}
where $\gamma$ is the solution to
\begin{align}
  \gamma \alpha (\log \alpha - 1) = - (1 + \log \gamma).
\end{align}

When $\alpha = \frac{3}{2}$, we have $\gamma \approx 1.7468$.

Therefore, the value of the objective function is

\begin{align}
  & \mu - \beta \left (\sum_{i_1 \in \Z} y_{i_1}^{(1)}  + \sum_{i_2 \in \Z} y_{i_2}^{(2)} + \cdots + \sum_{i_d \in \Z} y_{i_d}^{(d)} \right) \\
  = & \mu - \beta d  \sum_{i_1 \in \Z} y_{i_1}^{(1)} \\
  = & \mu - \beta d \left( \sum_{k' = 0}^{k-1} \sum_{j =0}^{\D-1}   y_{-(k'\D+j)}^{(1)} +  \sum_{j = 1}^{\D} \sum_{m=1}^{m(j)} y_{ (m-1)\D+j}^{(1)}    \right) \\
  = & \frac{d\D^2\log^2 \alpha}{\beta^2} - \beta d   \left( \frac{1-e^{-k\beta}}{1-e^{-\beta}} e^{k\beta} ( (z_1 k\D + z_2)\D + z_1 \frac{\D(\D-1)}{2} )  - z_1 \D^2 \frac{k(k-1)}{2} - z_1 k \frac{\D(\D-1)}{2} - z_2 k\D \right)  \nonumber \\
  &- \beta d  \sum_{j=1}^{\D}  (  \frac{e^\beta (1-e^{m(j)\beta})}{1-e^{\beta}} ( e^{k\beta}(z_1(k\D+\D-j)+z_2 ) - z_1(\D-j) \nonumber \\
  & - z_2 - z_3(\D-j) + z_4 )  - z_4 m(j) - z_3 \D \frac{m(j)(m(j)+1)}{2} + z_3(\D-j)m'   )  \\
  = & \frac{d\D^2}{\beta^2} \left(  \log^2 \alpha - (\alpha-1)(2\log\alpha -2) + \log^2 \alpha - 2 \log \alpha + (1-\gamma)\alpha (2\log \alpha -2) - 2\log \gamma - \log^2 \gamma    \right)  + o(\frac{1}{\beta^2} ) \\
  = & \frac{d\D^2}{\beta^2} \left(  2 \log^2 \alpha -2 - 2\alpha \gamma \log \alpha + 2\alpha \gamma - 2 \log \gamma - \log^2 \gamma  \right)   + o(\frac{1}{\beta^2} )  \\
  \approx & 0.0177 \frac{d\D^2}{\beta^2} + o(\frac{1}{\beta^2} ) .
\end{align}
as $\beta \triangleq \max(\e, \delta) \to 0$.

Therefore,
\begin{align}
  \lim_{\max(\e,\delta) \to 0 } \frac{V'_{LB}}{\frac{d\D^2}{\beta^2}} \ge   0.0177.
\end{align}


\end{proof}

\bibliographystyle{IEEEtran}

\bibliography{reference}

\end{document}